\documentclass{revtex4-2}
\usepackage[T1]{fontenc}
\usepackage{mathtools}
\usepackage{amssymb}
\usepackage{amsthm}
\usepackage{thmtools}
\usepackage{hyperref}
\usepackage{physics}
\usepackage{graphicx}
\usepackage{subfigure}
\usepackage{float}
\usepackage{bm}
\usepackage{stmaryrd}

\newtheorem{definition}{Definition}
\newtheorem{lemma}{Lemma}
\newtheorem{theorem}{Theorem}
\newtheorem{proposition}{Proposition}

\newtheorem{remark}{Remark}
\newtheorem{example}{Example}

\begin{document}
	
	\title{Construction and characterization of measures in block coherence resource theory}
	
	\author{Xiangyu Chen}
	\email{Electronic address: 23S012014@stu.hit.edu.cn}
	\affiliation{School of Mathematics, Harbin Institute of Technology, Harbin 150001, China}
	\author{Qiang Lei}
	\email{Electronic address: leiqiang@hit.edu.cn}
	\affiliation{School of Mathematics, Harbin Institute of Technology, Harbin 150001, China}
	
	\begin{abstract}
		Quantum coherence, as a direct manifestation of the quantum superposition principle, is a crucial resource in quantum information processing. Block coherence resource theory generalizes the traditional coherence framework by defining coherence via a set of orthogonal projectors. Within this framework, we investigates the construction and comparison of block coherence measures. First, we propose two universal methods for constructing coherence measures and introduce a two-parameter family of measures based on the $\alpha$-$z$ R\'enyi relative entropy and a family of measures based on the Tsallis relative operator entropy. Second, through theoretical proofs and numerical counterexamples, we compares the ordering relations and numerical magnitudes among different block coherence measures and establishes a series of universal numerical inequalities to constrain their values. Besides, we also use $C_{\alpha,1}$ to show the role of coherence in complex dynamic evolution of the Kominis master equation that includes recombination reactions.
	\end{abstract}
	
	\maketitle
	
	\section{Introduction}
	Quantum resource theory provides a systematic framework for quantifying and manipulating various physical resources in quantum information. The construction of resource theory originally stemmed from the quantitative study of entanglement \cite{vedral1998entanglement,Plenio2005AnIT}, with entanglement resource theory \cite{Hordecki2009Quantum} being an early successful example of this framework. Subsequently, this framework has been widely applied to quantum information fields, forming a unified research paradigm for quantum resource theory. In this context, an increasing number of physical concepts are viewed as resources. As a fundamental characteristic of the quantum state superposition principle, the standard coherence resource theory was established by Baumgratz et al. in 2014 \cite{baumgratz2014quantifying}. Complex numbers, as an indispensable element in describing quantum systems, led to the formal proposal of the imaginarity resource theory by Hickey et al. in 2018 \cite{hickey2018quantifying}. Additionally, there are other quantum resource theories such as magic \cite{Howard2017application}, thermodynamics \cite{Gour2015resource}, phase \cite{Xu2023quantifying} and so on \cite{Amaral2018Noncontextual,xu2019coherence,Luo2017partial} have been proposed. Together, these theories constitute a quantum resource research paradigm, aiming to reveal the fundamental value of various physical properties in information processing.
	\par
	Block coherence resource theory, as a generalization of standard coherence resource theory, can describe a more general form of coherence. It extends the set of free states from completely diagonal states to block-diagonal states, thereby effectively characterizing coherence preserved in systems with symmetries or subspace constraints. This theoretical framework was pioneered by \r Aberg et al. \cite{aberg2006quantifying} and formally defined by Bischof et al. Simultaneously, via Naimark extension, they embedded the system into a larger space, transforming the problem of POVM coherence into a block coherence problem in that space, thus establishing POVM coherence resource theory \cite{bischof2019resource,bischof2021quantifying}. Block coherence or POVM coherence resource theory contributes to research in cryptography \cite{bischof2021quantifying}, quantum measurement \cite{Yu2021quantum}, quantum biology \cite{kominis2025physiological} and other areas.
	\par
	In this paper, Section 2 elaborates the basic framework of block coherence resource theory and, based on convex combination and convex roof extension, proposes two universal methods for constructing block coherence measures. In section 3, utilizing the alternative framework, we introduce a two-parameter family of block coherence measures $C_{\alpha,z}$ based on the $\alpha$-$z$ R\'enyi relative entropy, and analyze a series of its properties. We also introduce a measure $C^N_{\beta}$ based on the Tsallis relative operator entropy. In section 4, we first reviews existing block coherence measures, and then investigates the ordering relations and numerical magnitude relations among different measures. We establishes a series of universal inequality constraints, and clarifies the intrinsic connections and applicable scopes of various measures. In section 5, we evaluate S-T coherence from the perspective of quantum resource theory. we quantitatively investigate the evolution of a sample state within the Kominis master equation that includes recombination reactions \cite{kominis2015radical}. Through numerical simulations, we show the role of coherence in complex dynamic processes and its influence on chemical reaction yields.
	\section{Block Coherence Resource Theory Framework and Construction of Block Coherence Measures}
	\subsection{Block Coherence Resource Theory Framework}
	\par
	In block coherence resource theory, a quantum state $\rho$ associated with a set of projectors $\mathbf{P}=\left\lbrace P_k \right\rbrace_k$ is called a free state (block-incoherent state) if it satisfies $\rho=\sum_kP_k\rho P_k$; otherwise, it is defined as a block-coherent state. The set of all free states is denoted by $\mathcal{I}_{\text{B}}$.
	\par
	A quantum channel $\Phi=\left\lbrace K_n \right\rbrace_n$ is called a free operation (block-incoherent operation) if and only if each of its Kraus operators satisfies, for any $\rho\in\mathcal{I}_{\text{B}}$:
	\begin{eqnarray*}
		P_iK_n\rho K_n^{\dagger}P_j=0,i\neq j.
	\end{eqnarray*}
	The set of all free operations is denoted by $\mathcal{IO}_{\text{B}}$. For any free operation, each Kraus operator $K_n$ has the following specific form \cite{bischof2021quantifying}:
	\begin{equation*}
		K_n=\sum_jP_{f_n(j)}M_nP_j.
	\end{equation*}
	Here, $f_n$ is a map from the index set $\left\lbrace 1,\cdots,n\right\rbrace $ to itself, and $M_n$ is some operator related to $n$ that ensures the equality holds.
	\par
	A proper block coherence measure should satisfy (B1) Non-negativity, (B2) Monotonicity, and (B3) Convexity. Later, in literature \cite{xu2020general}, condition (B4) Strong monotonicity was supplemented for measures in block coherence resource theory. Therefore, in block coherence resource theory, a block coherence measure $C\left( \rho ,\mathbf{P} \right)$ concerning the projector set $\mathbf{P}=\left\lbrace P_k \right\rbrace_k $ should satisfy the following conditions:
	\par
	\textbf{(B1)} Non-negativity: For any quantum state $\rho$, the measure $C$ should satisfy $C(\rho,\mathbf{P})\geqslant 0$, and $C(\rho,\mathbf{P})=0$ if and only if $\rho\in\mathcal{I}_{\text{B}}$;
	\par
	\textbf{(B2)} Monotonicity: For any quantum state $\rho$ and any free operation $\Phi$, the measure $C$ should satisfy $C(\Phi(\rho),\mathbf{P})\leqslant C(\rho,\mathbf{P})$;
	\par
	\textbf{(B3)} Strong monotonicity: For any quantum state $\rho $ and any free operation $\Phi=\left\lbrace K_n \right\rbrace_n \in\mathcal{IO}_{\text{B}}$, let $p_n=\text{Tr}(K_n\rho K_n^{\dagger})$ and $\rho_n=K_n\rho K_n^{\dagger}/p_n$, the measure $C$ should satisfy $\sum_np_nC(\rho_n,\mathbf{P})\leqslant C(\rho,\mathbf{P})$;
	\par
	\textbf{(B4)} Convexity: For any probability distribution $\left\lbrace p_n \right\rbrace_n$ satisfying $p_n\geqslant 0$ and $\sum_np_n=1$, and any set of quantum states $\left\lbrace \rho_n \right\rbrace_n$, the measure $C$ should satisfy $\sum_np_nC(\rho_n,\mathbf{P})\geqslant C(\sum_np_n\rho_n,\mathbf{P})$;
	\par
	Block coherence resource theory is an important generalization of the traditional (standard) coherence resource theory, and it also has close connections with others coherence resource theory.
	\par
	The traditional coherence resource theory (the framework proposed by Baumgratz et al. \cite{baumgratz2014quantifying}) can be viewed as a special case of block coherence theory. In the standard theory, free states (incoherent states) are defined as those diagonal in a fixed orthogonal computational basis $\{\ket{i}\}^d_{i=1}$. This basis can be equivalently described by a set of one-dimensional projectors $\{P_i=\ket{i}\bra{i}\}^d_{i=1}$. In this case, the block-dephasing operation $\Delta(\rho)=\sum_{i}P_i\rho P_i$ reduces to the ordinary dephasing (complete dephasing) operation. Therefore, when the given projector set $\mathbf{P}=\{P_i\}$, the block coherence theory completely reduces to the standard coherence theory. Block coherence theory generalizes the concept of coherence from the level of individual basis vectors to the level of subspaces ("blocks") by allowing projectors to define subspaces of arbitrary dimension, enabling a more natural description of coherence in complex systems.
	\par
	Partial coherence resource theory concerns the coherence of subsystem $A$ relative to its fixed computational basis $\{\ket{\nu_i}^A\}_{i=1}^{d_A}$ in a composite system $AB$, while subsystem $B$ can be in an arbitrary state. Its set of free states is $\mathcal{I}_{\text{P}}^A = \{\sigma^{AB} \in \mathcal{D}(\mathcal{H}^{AB}):\Pi_{\text{L}}(\sigma^{AB}) = \sigma^{AB} \}$. If we take the total Hilbert space of the composite system $AB$ and define the projectors as the Lüders measurement corresponding to the computational basis of system A, i.e., $\Pi_{\text{L}}=\{\Pi_i^A\}_{i=1}^{d_A}=\{\ket{\nu_i}\bra{\nu_i}^A\otimes I^B\}_{i=1}^{d_A}$, then the block coherence free state set defined thereby is identical to the partial coherence free state set. Therefore, partial coherence theory is a special case of block coherence theory where the projectors act on a subsystem $A$ of a composite system, and can be viewed as an application of block coherence theory in a specific scenario.
	\par
	Furthermore, when the dimension of subsystem $B$ is 1 (i.e., $B$ is a trivial one-dimensional system), the state space of the composite system $AB$ is physically equivalent to that of subsystem $A$. In this case, the free state in partial coherence theory, $\sigma=\sum_{i=1}^{d_A}\ket{\nu_i}\bra{\nu_i}^A\otimes\tau_i^B$, degenerates (since $\tau_i^B$ is a scalar) to $\sigma=\sum_{i=1}^{d_A}p_i\ket{\nu_i}\bra{\nu_i}^A$, which is precisely the standard incoherent state on subsystem A. Simultaneously, the projectors $\Pi_{\text{L}}=\{\Pi_i^A\}_{i=1}^{d_A}=\{\ket{\nu_i}\bra{\nu_i}^A\otimes I^B\}_{i=1}^{d_A}$ degenerate to $\{\ket{\nu_i}\bra{\nu_i}^A\}_{i=1}^{d_A}$. Hence, when subsystem $B$ has dimension 1, partial coherence theory further reduces to standard coherence theory.
	\par
	Singlet-triplet (S-T) coherence is a focus of extensive research in quantum biology. The S-T coherence studied by Kominis \cite{kritsotakis2014retrodictive,kominis2015radical,kominis2020quantum,kominis2025physiological} can be directly related to block coherence resource theory. We can term it S-T coherence resource theory, which serves as a typical and important concrete instance of block coherence resource theory. In quantum biology, Radicals are particles possessing unpaired valence electrons. Radicals or radical pairs play a very important role in thermal reactions, radiation reactions, and photochemical reactions. The radical state can be represented by an electron oscillation state and a nuclear state. The two unpaired electron spins give rise to the singlet and triplet states, respectively expressed as:
	\begin{align*}
		\ket{S}&=\frac{1}{\sqrt{2}}(\ket{\uparrow\downarrow}-j\ket{\downarrow\uparrow}),
		\\
		\ket{T_{-1}}&=\ket{\downarrow\downarrow},
		\\
		\ket{T_{0}}&=\frac{1}{\sqrt{2}}(\ket{\uparrow\downarrow}+\ket{\downarrow\uparrow}),
		\\
		\ket{T_{1}}&=\ket{\uparrow\uparrow},
	\end{align*}
	where $\ket{\uparrow}$ and $\ket{\downarrow}$ represent electron spin up and spin down, respectively, and can be mathematically viewed as:
	\begin{eqnarray*}
		\ket{\uparrow}=\begin{pmatrix}
			1 \\ 0
		\end{pmatrix},
		\ket{\downarrow}=\begin{pmatrix}
			0 \\ 1
		\end{pmatrix}.
	\end{eqnarray*}
	\par
	Due to a series of biochemical reactions, the radical state oscillates between the singlet and triplet states, resulting in the electron oscillation state. Neglecting the nuclear state yields a four-dimensional Hilbert space spanned by the singlet and triplet states. The projectors $Q_S=\frac{1}{4}I-s_Ds_A$ and $Q_T=\frac{3}{4}I+s_Ds_A$ correspond to the singlet subspace and the triplet subspace, respectively, where $s_D$ and $s_A$ are the electron spins of the donor and acceptor, respectively.
	\par
	The $j$-th component of the donor spin $s_D$ and the $k$-th component of the acceptor spin $s_A$ can be represented as:
	\begin{align*}
		s_{jD}&=\hat{s}_j\otimes I_D,
		\quad
		s_{kA}&=I_A\otimes\hat{s}_k,
	\end{align*}
	where $\hat{s}_j$ is the $j$-th component of the donor's electron spin, and the identity operator $I_D$ corresponds to the $j$-th component of the acceptor's electron spin; the identity operator $I_A$ corresponds to the $k$-th component of the donor's electron spin, and $\hat{s}_k$ is the $k$-th component of the acceptor's electron spin. Here, $\hat{s}_l=\frac{\hbar}{2}\sigma_l$, $l=x,y,z$, and $\sigma_l$ are the Pauli operators:
	\begin{eqnarray*}
		\sigma_x=\begin{pmatrix}
			0 & 1 \\
			1 & 0
		\end{pmatrix},
		\sigma_y=\begin{pmatrix}
			0 & -i \\
			i & 0
		\end{pmatrix},
		\sigma_z=\begin{pmatrix}
			1 & 0 \\
			0 & -1
		\end{pmatrix},
	\end{eqnarray*}
	Then, in natural units (setting $\hbar=1$), we obtain:
	\begin{eqnarray*}
		s_Ds_A=\frac{1}{4}\begin{pmatrix}
			1 & & & \\
			& -1 & 2 & \\
			& 2 & -1 & \\
			& & & 1
		\end{pmatrix}.
	\end{eqnarray*}
	Thus, it is known that the projectors are $Q_S=\ket{S}\bra{S}$ and $Q_T=\ket{T_1}\bra{T_1}+\ket{T_0}\bra{T_0}+\ket{T_{-1}}\bra{T_{-1}}$. In this paper, for notational convenience, we sometimes denote $\rho_{ij}=Q_i\rho Q_j$, where $i,j\in\left\lbrace S,T\right\rbrace $.
	\par 
	In the radical pair model, the system's Hilbert space is partitioned into the singlet subspace and the triplet subspace. This corresponds to a complete set of orthogonal projectors defined by two operators: $\mathbf{Q}=\{Q_S,Q_T\}$, where $Q_S=\ket{S}\bra{S}$ projects onto the singlet subspace, and $Q_T=\sum_{m=-1}^{1}\ket{T_m}\bra{T_m}$ projects onto the triplet subspace. The S-T coherence free state set defined thereby, $\mathcal{I}_{\text{ST}}=\{\sigma:\sigma=\sigma_{SS}+\sigma_{TT}\}$, is precisely the specific form of the block coherence free state set when taking $\mathbf{P}=\mathbf{Q}$. Therefore, S-T coherence resource theory is a direct application of block coherence resource theory in the specific case where the projector set $\mathbf{P}$ contains only two projectors corresponding to the singlet and triplet subspaces, respectively.
	\par
	The relationships among the various coherence resource theories mentioned above are illustrated in the following diagram:
	\begin{figure}[H]
		\centering
		\includegraphics[width=0.5\linewidth,height=0.255\textheight]{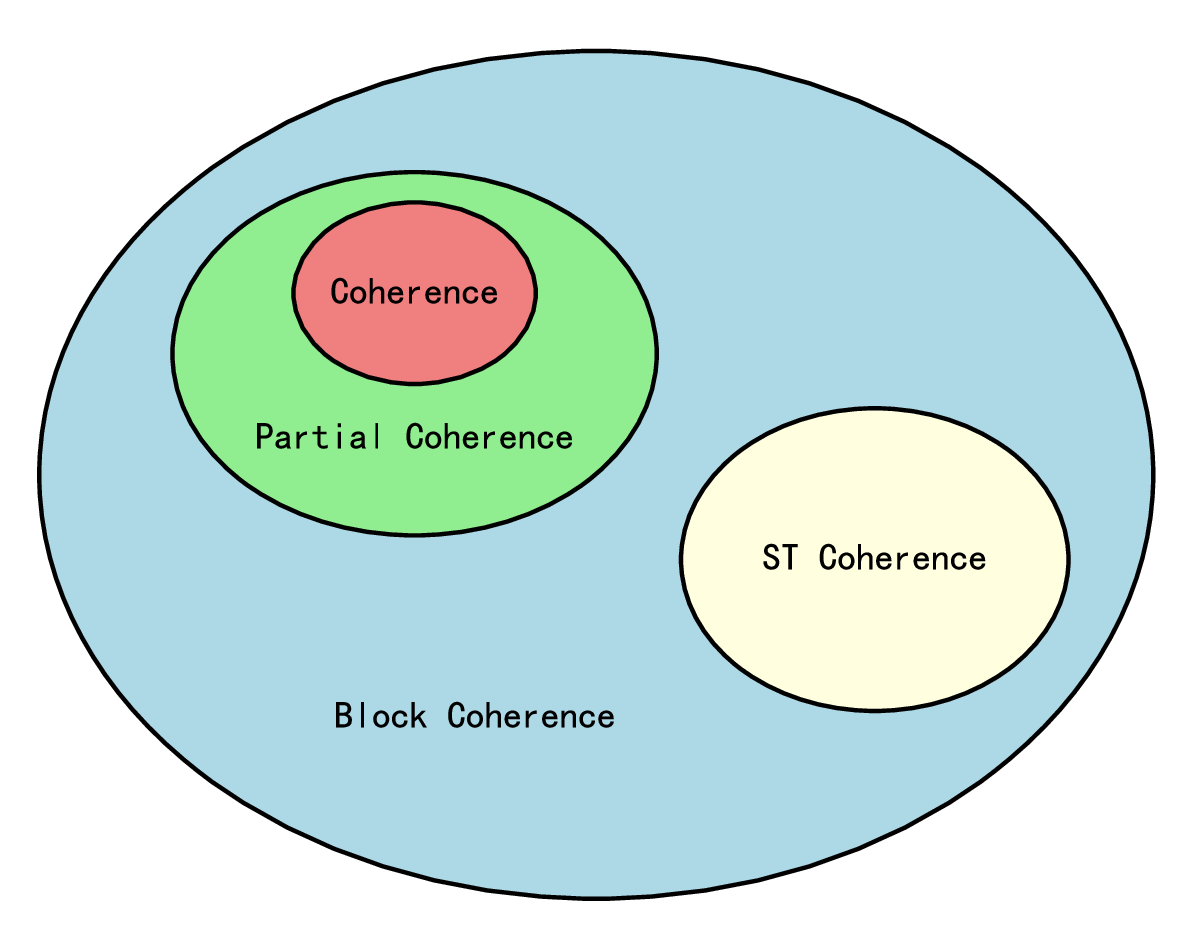}
		\caption{Relationships among various coherence resource theories.}
	\end{figure}
	\par
	It is evident that block coherence resource theory establishes a more general and inclusive framework for describing quantum coherence. This greatly expands the scope and applicability of quantum coherence research, enabling it to be applied to a wider range of physical systems, such as many-body systems with specific symmetries or quantum biology models. The subsequent construction and comparative study of measures in this work are all conducted within this general framework. Having clarified the basic conditions for a block coherence measure, a natural question arises: how can we construct new, equally valid block coherence measures from known ones? Next, we will construct block coherence measures from different perspectives.

	\subsection{Construction of Block Coherence Measures via Convex Combination and Convex Roof Extension}
	\par
	The following theorem demonstrates that the set of block coherence measures is closed under convex combination, meaning a new measure constructed in this way still satisfies (B1)-(B4). This provides a convex combination method for generating new block coherence measures by weighting different existing ones.
	\begin{theorem}
		For any set of block coherence measures $\left\lbrace C_j(\rho,\mathbf{P}) \right\rbrace_j$ associated with the projector set $\mathbf{P}=\left\lbrace P_k \right\rbrace_k$, and any set of constants $\left\lbrace q_j \right\rbrace_j$ satisfying $q_j\geqslant 0$ and $\sum_jq_j=1$, $C(\rho,\mathbf{P}):=\sum_jq_jC_j(\rho,\mathbf{P})$ is a block coherence measure.
	\end{theorem}
	\begin{proof}
		\par\noindent\par
		We prove the conditions (B1)-(B4) in sequence.
		\par\noindent
		(B1) Non-negativity: For any quantum state $\rho$, each block coherence measure $C_j$ satisfies $C_j(\rho,\mathbf{P})\geqslant 0$, hence $C(\rho,\mathbf{P})\geqslant 0$. If $C(\rho,\mathbf{P})=0$, since $q_j\geqslant 0$, then each $C_j(\rho,\mathbf{P})=0$, implying $\rho\in\mathcal{I}_{\text{B}}$.
		\par\noindent
		(B2) Monotonicity: For any quantum state $\rho $ and any free operation $\Phi$, we have
		\begin{align*}
			C(\rho,\mathbf{P})=\sum_jq_jC_j(\rho,\mathbf{P})\geqslant\sum_jq_jC_j(\Phi(\rho),\mathbf{P})=C(\Phi(\rho),\mathbf{P}).
		\end{align*}
		\par\noindent
		(B3) Strong monotonicity: For any quantum state $\rho $ and any free operation $\Phi=\left\lbrace K_n \right\rbrace_n$, let $p_n=\text{Tr}(K_n\rho K_n^{\dagger})$ and $\rho_n=K_n\rho K_n^{\dagger}/p_n$. Since each $C_j$ satisfies $\sum_np_nC_j(\rho_n,\mathbf{P})\leqslant C_j(\rho,\mathbf{P})$, it follows that
		\begin{align*}
			C(\rho,\mathbf{P})
			&=\sum_jq_jC_j(\rho,\mathbf{P})
			\\
			&\geqslant\sum_jq_j\sum_np_nC_j(\rho_n,\mathbf{P})
			\\
			&=\sum_np_n\sum_jq_jC_j(\rho_n,\mathbf{P})
			\\
			&=\sum_np_nC(\rho_n,\mathbf{P}).
		\end{align*}
		\par\noindent
		(B4) Convexity: For any probability distribution $\left\lbrace p_n \right\rbrace_n$ satisfying $p_n\geqslant 0$ and $\sum_np_n=1$, and any set of quantum states $\left\lbrace \rho_n \right\rbrace_n$, we have
		\begin{align*}
			C\left(\sum_np_n\rho_n,\mathbf{P}\right)
			&=\sum_jq_jC_j\left(\sum_np_n\rho_n,\mathbf{P}\right)
			\\
			&\leqslant\sum_jq_j\sum_np_nC_j(\rho_n,\mathbf{P})
			\\
			&=\sum_np_n\sum_jq_jC_j(\rho_n,\mathbf{P})
			\\
			&=\sum_np_nC(\rho_n,\mathbf{P})
		\end{align*}
		This completes the proof that $C$ is a block coherence measure.
	\end{proof}
	\par
	Theorem 3.1 illustrates the approach of constructing new measures via convex combination of existing ones. Furthermore, we can focus on a measure itself and define a new measure through an optimization procedure. Specifically, we can consider minimizing over all possible decompositions of a quantum state, as shown in the next theorem. This method is known as the convex roof extension in quantum resource theory \cite{chen2023measures,liu2018superadditivity}.
	\begin{theorem}
		For a block coherence measure $C(\rho,\mathbf{P})$,
		\begin{eqnarray*}
			\hat{C}(\rho,\mathbf{P}):=\min_{\left\lbrace q_j,\rho_j\right\rbrace }\sum_jq_jC(\rho_j,\mathbf{P})
		\end{eqnarray*}
		is a block coherence measure. Here, $\min_{\left\lbrace q_j,\rho_j\right\rbrace }$ denotes the minimum value obtained by summing over all possible decompositions $\rho=\sum_jq_j\rho_j$ of $\rho$, where $q_j\geqslant 0$ and $\sum_jq_j=1$.
	\end{theorem}
	\begin{proof}
		\par\noindent\par\noindent
		(B1) Non-negativity: Since each $\rho_j$ yields $C(\rho_j,\mathbf{P})\geqslant 0$, we have $\hat{C}(\rho,\mathbf{P})\geqslant 0$. If $\hat{C}(\rho,\mathbf{P})=0$, then there exists a decomposition $\rho=\sum_jq_j\rho_j$ such that $\sum_jq_jC(\rho_j,\mathbf{P})=0$. Given $q_j\geqslant 0$, this implies $C(\rho_j,\mathbf{P})=0$ for each $j$, meaning each $\rho_j\in\mathcal{I}_{\text{B}}$. Consequently,
		\begin{eqnarray*}
			\rho=\sum_jq_j\rho_j=\sum_jq_j\sum_kP_k\rho_jP_k=\sum_kP_k(\sum_jq_j\rho_j)P_k=\sum_kP_k\rho P_k
		\end{eqnarray*}
		that is, $\rho\in\mathcal{I}_{\text{B}}$.
		\par\noindent
		(B2) Monotonicity: There exists a decomposition $\rho=\sum_jq_j\rho_j$ that achieves the optimum for $\sum_jq_jC(\rho_j,\mathbf{P})$. Since $\Phi(\rho)=\sum_jq_j\Phi(\rho_j)$, the set $\left\lbrace q_j,\Phi(\rho_j) \right\rbrace $ constitutes a decomposition of $\Phi(\rho)$. Therefore, $\hat{C}(\rho,\mathbf{P})=\sum_jq_jC(\rho_j,\mathbf{P})\geqslant \sum_jq_j C(\Phi(\rho_j),\mathbf{P})\geqslant\hat{C}(\Phi(\rho),\mathbf{P})$.
		\par\noindent
		(B3) Strong monotonicity: First, for any set $\left\lbrace q_j\right\rbrace_j$ satisfying $q_j\geqslant 0$ and $\sum_jq_j=1$, it always holds that $\rho=\sum_jq_j\rho$. By the definition of $\hat{C}$, we have
		\begin{eqnarray*}
			\hat{C}(\rho,\mathbf{P}):=\min_{\left\lbrace q_j,\rho_j\right\rbrace }\sum_jq_jC(\rho_j,\mathbf{P})\leqslant\sum_jq_jC(\rho,\mathbf{P})=C(\rho,\mathbf{P}).
		\end{eqnarray*}
		Take the optimal decomposition $\rho=\sum_jq_j\rho_j$. For a free operation $\Phi=\left\lbrace K_n \right\rbrace_n$, let $p_n=\text{Tr}(K_n\rho K_n^{\dagger})$ and $\rho_n=K_n\rho K_n^{\dagger}/p_n$. Denote $a_{j,n}=\text{Tr}(K_n\rho_j K_n^{\dagger})$. Then we have
		\begin{eqnarray*}
			\sum_jq_ja_{j,n}=\sum_jq_j\text{Tr}(K_n\rho_j K_n^{\dagger})=\text{Tr}[K_n(\sum_jq_j\rho_j)K_n^{\dagger}]=\text{Tr}(K_n\rho K_n^{\dagger})=p_n.
		\end{eqnarray*}
		Thus, the following holds:
		\begin{align*}
			\hat{C}(\rho,\mathbf{P})&=\sum_jq_jC(\rho_j,\mathbf{P})
			\\
			&\geqslant\sum_jq_j\sum_na_{j,n}C\left( \frac{K_n\rho_j K_n^{\dagger}}{a_{j,n}},\mathbf{P}\right)
			\\
			&=\sum_n p_n\sum_j\frac{q_ja_{j,n}}{p_n}C\left( \frac{K_n\rho_j K_n^{\dagger}}{a_{j,n}},\mathbf{P}\right)
			\\
			&\geqslant\sum_n p_n C\left(\sum_j\frac{q_ja_{j,n}}{p_n}\frac{K_n\rho_j K_n^{\dagger}}{a_{j,n}} ,\mathbf{P}\right)
			\\
			&=\sum_np_nC\left( \frac{K_n\rho K_n^{\dagger}}{p_n},\mathbf{P}\right)
			\\
			&=\sum_np_nC(\rho_n,\mathbf{P})
			\\
			&\geqslant\sum_np_n\hat{C}(\rho_n,\mathbf{P}).
		\end{align*}
		The first inequality holds due to the strong monotonicity of $C(\rho,\mathbf{P})$, and the second inequality holds due to the convexity of $C(\rho,\mathbf{P})$.
		\par\noindent
		(B4) Convexity: To simplify the exposition, we prove convexity by considering only two states $\rho_1$ and $\rho_2$, without loss of generality. Let $\rho_1=\sum_mq_{1,m}\rho_{1,m}$ and $\rho_2=\sum_nq_{2,n}\rho_{2,n}$ be the optimal decompositions for $\rho_1$ and $\rho_2$, respectively. For any constants $a\geqslant 0$ and $b\geqslant 0$ satisfying $a+b=1$, let $\rho=a\rho_1+b\rho_2$. Then
		\begin{eqnarray*}
			a\hat{C}(\rho_1,\mathbf{P})+b\hat{C}(\rho_2,\mathbf{P})=\sum_maq_{1,m}C(\rho_{1,m},\mathbf{P})+\sum_nbq_{2,n}C(\rho_{2,n},\mathbf{P}).
		\end{eqnarray*}
		For notational convenience, assuming $m\neq n$, let us suppose $m>n$. We supplement the set $\left\lbrace q_{2,n},\rho_{2,n} \right\rbrace $ with $m-n$ element pairs $\left\lbrace 0,\bf{0}\right\rbrace $, so that it contains $m$ elements. Replace $\left\lbrace aq_{1,m},\rho_{1,m}\right\rbrace$ and $\left\lbrace bq_{2,m},\rho_{2,m}\right\rbrace$ with $\left\lbrace p_{2k-1},\rho_{2k-1} \right\rbrace $ and $\left\lbrace p_{2k},\rho_{2k} \right\rbrace $, respectively. Then we have:
		\begin{align*}
			a\hat{C}(\rho_1,\mathbf{P})+b\hat{C}(\rho_2,\mathbf{P})
			&=\sum_kp_{2k-1}C(\rho_{2k-1},\mathbf{P})+p_{2k}C(\rho_{2k},\mathbf{P})
			\\
			&\geqslant\min_{\left\lbrace q_j,\rho_j\right\rbrace }\sum_jq_jC(\rho_j,\mathbf{P})
			\\
			&=\hat{C}(\rho,\mathbf{P})
			\\
			&=\hat{C}(a\rho_1+b\rho_2,\mathbf{P}).
		\end{align*}
		From the above, it follows that $\hat{C}(\rho,\mathbf{P})$ is a block coherence measure.
	\end{proof}
	\par
	Theorems 3.1 and 3.2 provide general algebraic methods for constructing new measures from known ones. However, as can be seen, condition (B3) strong monotonicity is often not trivial to prove. Therefore, analogous to the alternative framework in coherence resource theory \cite{yu2016alternative} and imaginarity resource theory \cite{xue2021quantification}, literature \cite{xu2020general} introduces condition (B5) block additivity as follows, which makes conditions (B1)-(B4) equivalent to conditions (B1), (B2), and (B5).
	\par
	\textbf{(B5)} Block Additivity: Arbitrarily partition the projector set $\mathbf{P}$ such that $\mathbf{P}=\left\lbrace P_{k_1} \right\rbrace_{k_1}\cup\left\lbrace P_{k_2} \right\rbrace_{k_2}$, where $\left\lbrace k_1 \right\rbrace_{k_1}\cup\left\lbrace k_2 \right\rbrace_{k_2}=\left\lbrace k \right\rbrace_k$ and $\left\lbrace k_1 \right\rbrace_{k_1}\cap\left\lbrace k_2 \right\rbrace_{k_2}=\emptyset$. For any $p_1>0$ and $p_2>0$ satisfying $p_1+p_2=1$, and any quantum states $\rho_1$ and $\rho_2$ satisfying $\rho_1P_{k_2}=\rho_2P_{k_1}=0$, the measure $C$ must satisfy $C(p_1\rho_1\oplus p_2\rho_2,\mathbf{P})=p_1C(\rho_1,\mathbf{P})+p_2C(\rho_2,\mathbf{P})$.
	\par
	With the aid of this alternative framework, we can construct and verify new families of block coherence measures based on more general entropy functions. In the next section, we will present a family of measures defined via the $\alpha-z$ Rényi relative entropy.
	
	\section{Construction of block coherence measures via alternative framework}
	\subsection{Block coherence measures induced by $\alpha$-$z$ Rényi Relative Entropy}
	The $\alpha$-$z$ Rényi relative entropy, which possesses two parameters $\alpha$ and $z$, has excellent properties and serves as a generalization of various relative entropies. The $\alpha$-$z$ Rényi relative entropy is defined in literature \cite{audenaert2015alpha} as:
	\begin{eqnarray*}
		D_{\alpha,z}(A||B)=\frac{1}{\alpha-1}\log\text{Tr}(B^{\frac{1-\alpha}{2z}}A^{\frac{\alpha}{z}}B^{\frac{1-\alpha}{2z}})^z.
		\nonumber
	\end{eqnarray*}
	\par
	To construct and prove the new block coherence measure to be presented next, we first need to introduce two key matrix inequalities as lemmas \cite{audenaert2007araki}:
	\begin{lemma}
		Let $A$ and $B$ be two positive semidefinite matrices. For $r\geqslant1$, and $q\geqslant 0$, the following inequality holds:
		\begin{align*}
			\text{Tr}(ABA)^{rq}\leqslant\text{Tr}(A^rB^rA^r)^q.
		\end{align*}
	\end{lemma}
	\par
	Utilizing the above lemma, we can now define a new family of block coherence measures based on the $\alpha$-$z$ Rényi relative entropy and verify that it satisfies the required conditions for a measure. We have the following theorem:
	\begin{theorem}
		The function $C_{\alpha,z}$ defined via the $\alpha$-$z$ Rényi relative entropy is a valid block coherence measure, defined as:
		\begin{eqnarray*}
			C_{\alpha,z}(\rho,\mathbf{P})=1-\max_{\sigma\in\mathcal{I}_{\text{B}}}\left\lbrace \text{Tr}(\sigma^{\frac{1-\alpha}{2z}}\rho^{\frac{\alpha}{z}}\sigma^{\frac{1-\alpha}{2z}})^z \right\rbrace^{\frac{1}{\alpha}},
		\end{eqnarray*}
		where $\alpha\in(0,1)$ and $\text{max}\left\lbrace \alpha,1-\alpha\right\rbrace \leqslant z$.
	\end{theorem}
	\begin{proof}
		\par\noindent\par
		Condition (B3) is not convenient to prove directly, so we complete this proof via the alternative framework.
		\par\noindent
		(B1) Non-negativity: From literature \cite{muller2013quantum}, we know $D_{\alpha,\alpha}(\rho||\sigma)\geqslant 0$. This implies $\log\text{Tr}(\sigma^{\frac{1-\alpha}{2\alpha}}\rho^{\frac{\alpha}{z}}\sigma^{\frac{1-\alpha}{2\alpha}})^\alpha\leqslant 0$, i.e., $\text{Tr}(\sigma^{\frac{1-\alpha}{2\alpha}}\rho^{\frac{\alpha}{\alpha}}\sigma^{\frac{1-\alpha}{2\alpha}})^\alpha\leqslant 1$. Then, using Lemma 3.1, we can obtain:
		\begin{align*}
			\text{Tr}(\sigma^{\frac{1-\alpha}{2z}}\rho^{\frac{\alpha}{z}}\sigma^{\frac{1-\alpha}{2z}})^z=\text{Tr}(\sigma^{\frac{1-\alpha}{2z}}\rho^{\frac{\alpha}{z}}\sigma^{\frac{1-\alpha}{2z}})^{\frac{z}{\alpha}\alpha}\leqslant\text{Tr}(\sigma^{\frac{1-\alpha}{2\alpha}}\rho^{\frac{\alpha}{\alpha}}\sigma^{\frac{1-\alpha}{2\alpha}})^\alpha\leqslant 1.
		\end{align*}
		Thus, $\max_{\sigma\in\mathcal{I}_{\text{B}}}\left\lbrace \text{Tr}(\sigma^{\frac{1-\alpha}{2z}}\rho^{\frac{\alpha}{z}}\sigma^{\frac{1-\alpha}{2z}})^z \right\rbrace^{\frac{1}{\alpha}}\leqslant 1$, which yields $C_{\alpha,z}(\rho,\mathbf{P})\geqslant 0$. When $C_{\alpha,z}(\rho,\mathbf{P})=0$, there exists $\gamma\in\mathcal{I}_{\text{B}}$ such that $\text{Tr}(\gamma^{\frac{1-\alpha}{2z}}\rho^{\frac{\alpha}{z}}\gamma^{\frac{1-\alpha}{2z}})^z=1$. Then, according to literature \cite{audenaert2015alpha}, the following equivalence chain holds:
		\begin{eqnarray*}
			C_{\alpha,z}(\rho,\mathbf{P})=0&\Leftrightarrow&\text{Tr}(\gamma^{\frac{1-\alpha}{2z}}\rho^{\frac{\alpha}{z}}\gamma^{\frac{1-\alpha}{2z}})^z=1
			\\
			&\Leftrightarrow&\log\text{Tr}(\gamma^{\frac{1-\alpha}{2z}}\rho^{\frac{\alpha}{z}}\gamma^{\frac{1-\alpha}{2z}})^z=0
			\\
			&\Leftrightarrow&D_{\alpha,z}(\rho||\gamma)=0
			\\
			&\Leftrightarrow&\rho=\gamma
		\end{eqnarray*}
		Therefore, $C_{\alpha,z}(\rho,\mathbf{P})=0$ is equivalent to $\rho\in\mathcal{I}_{\text{B}}$, which proves non-negativity.
		\par\noindent
		(B2) Monotonicity: Literature \cite{audenaert2015alpha} provides the DPI inequality for $D_{\alpha,z}$, namely $D_{\alpha,z}[\Phi(\rho)||\Phi(\sigma)]\leqslant D_{\alpha,z}(\rho||\sigma)$. This is equivalent to $\text{Tr}[\Phi(\sigma)^{\frac{1-\alpha}{2z}}\Phi(\rho)^{\frac{\alpha}{z}}\Phi(\sigma)^{\frac{1-\alpha}{2z}}]^z\geqslant\text{Tr}(\sigma^{\frac{1-\alpha}{2z}}\rho^{\frac{\alpha}{z}}\sigma^{\frac{1-\alpha}{2z}})^z$. Thus, there exists $\gamma\in\mathcal{I}_{\text{B}}$ such that:
		\begin{align*}
			C_{\alpha,z}(\rho,\mathbf{P})
			&=
			1-\left\lbrace \text{Tr}(\gamma^{\frac{1-\alpha}{2z}}\rho^{\frac{\alpha}{z}}\gamma^{\frac{1-\alpha}{2z}})^z \right\rbrace^{\frac{1}{\alpha}}
			\\
			&\geqslant
			1-\left\lbrace \text{Tr}[\Phi(\gamma)^{\frac{1-\alpha}{2z}}\Phi(\rho)^{\frac{\alpha}{z}}\Phi(\gamma)^{\frac{1-\alpha}{2z}}]^z \right\rbrace^{\frac{1}{\alpha}}
			\\
			&\geqslant
			1-\max_{\sigma\in\mathcal{I}_{\text{B}}}\left\lbrace \text{Tr}[\sigma^{\frac{1-\alpha}{2z}}\Phi(\rho)^{\frac{\alpha}{z}}\sigma^{\frac{1-\alpha}{2z}}]^z \right\rbrace^{\frac{1}{\alpha}}
			\\
			&=
			C_{\alpha,z}(\Phi(\rho),\mathbf{P}).
		\end{align*}
		This proves monotonicity.
		\par\noindent
		(B5) Block Additivity: In the requirement of block additivity, let $\rho=p_1\rho_1\oplus p_2\rho_2$. For any $\sigma\in\mathcal{I}_{\text{B}}$, there exists a decomposition such that $\sigma=q_1\sigma_1\oplus q_2\sigma_2$, where the coefficients $q_1>0$ and $q_2>0$ are undetermined constants satisfying $q_1+q_2=1$, and $\dim\sigma_j=\dim\rho_j$. Therefore, we know:
		\begin{align*}
			&\hspace{1.1em}\max_{\sigma} \text{Tr}(\sigma^{\frac{1-\alpha}{2z}}\rho^{\frac{\alpha}{z}}\sigma^{\frac{1-\alpha}{2z}})^z
			\\
			&=
			\max_{q_1, q_2}\left\lbrace q_1^{1-\alpha}p_1^{\alpha}\max_{\sigma_1} \text{Tr}(\sigma_1^{\frac{1-\alpha}{2z}}\rho_1^{\frac{\alpha}{z}}\sigma_1^{\frac{1-\alpha}{2z}})^z+q_2^{1-\alpha}p_2^{\alpha}\max_{\sigma_2} \text{Tr}(\sigma_2^{\frac{1-\alpha}{2z}}\rho_2^{\frac{\alpha}{z}}\sigma_2^{\frac{1-\alpha}{2z}})^z\right\rbrace
			\\
			&=
			\max_{q_1, q_2}(q_1^{1-\alpha}p_1^{\alpha}t_1+q_2^{1-\alpha}p_2^{\alpha}t_2),
		\end{align*}
		where $t_i=\max_{\sigma_i} \text{Tr}(\sigma_i^{\frac{1-\alpha}{2z}}\rho_i^{\frac{\alpha}{z}}\sigma_i^{\frac{1-\alpha}{2z}})^z$.
		\par
		Next, using the H{\"o}lder inequality, we have:
		\begin{align*}
			q_1^{1-\alpha}p_1^{\alpha}t_1+q_2^{1-\alpha}p_2^{\alpha}t_2&\leqslant \left[ \sum_{i}(q_i^{1-\alpha})^{\frac{1}{1-\alpha}}\right] ^{1-\alpha}\left[ \sum_{i}(p_i^{\alpha}t_i)^{\frac{1}{\alpha}}\right] ^{\alpha}
			\\
			&=\left(\sum_{i}q_i\right) ^{1-\alpha}\left[ \sum_{i}(p_i^{\alpha}t_i)^{\frac{1}{\alpha}}\right] ^{\alpha}
			\\
			&=\left[ \sum_{i}(p_i^{\alpha}t_i)^{\frac{1}{\alpha}}\right] ^{\alpha},
		\end{align*}
		The inequality becomes an equality if and only if $q_i=cp_it_i^{\frac{1}{\alpha}}$, where $c=(p_1t_1^{\frac{1}{\alpha}}+p_2t_2^{\frac{1}{\alpha}})^{-1}$. Thus, the originally undetermined $q_1$ and $q_2$ can be determined by $p_j$ and $t_j$. We have:
		\begin{eqnarray*}
			\max_{\sigma}\left\lbrace  \text{Tr}(\sigma^{\frac{1-\alpha}{2z}}\rho^{\frac{\alpha}{z}}\sigma^{\frac{1-\alpha}{2z}})^z\right\rbrace^{\frac{1}{\alpha}}=\sum_{i}(p_i^{\alpha}t_i)^{\frac{1}{\alpha}}=\sum_{i}p_i\max_{\sigma_i}\left\lbrace  \text{Tr}(\sigma_i^{\frac{1-\alpha}{2z}}\rho_i^{\frac{\alpha}{z}}\sigma_i^{\frac{1-\alpha}{2z}})^z\right\rbrace ^{\frac{1}{\alpha}}.
		\end{eqnarray*}
		This implies $C_{\alpha,z}(\rho,\mathbf{P})=p_1C_{\alpha,z}(\rho_1,\mathbf{P})+p_2C_{\alpha,z}(\rho_2,\mathbf{P})$, i.e., additivity holds.
		\par
		So far, we have verified conditions (B1), (B2), and (B5) to confirm that $C_{\alpha,z}(\rho,\mathbf{P})$ is a valid block coherence measure.
	\end{proof}
	\par
	By verifying the alternative framework conditions, we have established the validity of $C_{\alpha,z}$ as a block coherence measure. An important value of the $C_{\alpha,z}$ family lies in the fact that by choosing specific parameters $\alpha$ and $z$, it reduces to or connects with many known measures, thereby forming a more unified framework for studying block coherence. We can provide two examples to illustrate:
	\begin{example}
		When $\alpha=0.5$ and $z=0.5$, $C_{\alpha,z}$ becomes:
		\begin{align*}
			C_{0.5,0.5}(\rho,\mathbf{P})
			&=1-\max_{\sigma\in\mathcal{I}_{\text{B}}}\left( \text{Tr}\sqrt{\sqrt{\rho}\sigma\sqrt{\rho}} \right)^2
			\\
			&=1-\max_{\sigma\in\mathcal{I}_{\text{B}}}F^2(\rho,\sigma)
			\\
			&=1-\max_{\sigma\in\mathcal{I}_{\text{B}}}F^2(\sigma,\rho)
			\\
			&=C_{geo}(\rho,\mathbf{P}).
		\end{align*}
		This is precisely the geometric block coherence measure.
	\end{example}
	\begin{example}
		When $\alpha=0.5$ and $z=1$, to compute $C_{\alpha,z}$, we need a useful alternative expression. In literature \cite{xu2020general}, it is proven that for any $\alpha>0$:
		\begin{eqnarray*}
			\frac{\left\lbrace \sum_i\text{Tr}[(P_i\rho^\alpha P_i)^{\frac{1}{\alpha}}] \right\rbrace^{\alpha}-1 }{\alpha-1}=\min_{\sigma\in\mathcal{I}_{\text{B}}}\frac{\text{Tr}(\rho^{\alpha}\sigma^{1-\alpha})-1}{\alpha-1},
		\end{eqnarray*}
		thus it holds that:
		\begin{eqnarray*}
			\sum_i\text{Tr}[(P_i\rho^\alpha P_i)^{\frac{1}{\alpha}}]=\left\lbrace \min_{\sigma\in\mathcal{I}_{\text{B}}}\text{Tr}(\rho^{\alpha}\sigma^{1-\alpha}) \right\rbrace ^{\frac{1}{\alpha}}=\min_{\sigma\in\mathcal{I}_{\text{B}}}[\text{Tr}(\rho^{\alpha}\sigma^{1-\alpha})]^{\frac{1}{\alpha}}.
		\end{eqnarray*}
		Therefore, when $z=1$ we have:
		\begin{align*}
			C_{\alpha,1}(\rho,\mathbf{P})
			&=1-\max_{\sigma\in\mathcal{I}_{\text{B}}}\left\lbrace \text{Tr}(\sigma^{\frac{1-\alpha}{2}}\rho^{\alpha}\sigma^{\frac{1-\alpha}{2}}) \right\rbrace^{\frac{1}{\alpha}}
			\\
			&=1-\max_{\sigma\in\mathcal{I}_{\text{B}}}\left\lbrace \text{Tr}(\rho^{\alpha}\sigma^{1-\alpha}) \right\rbrace^{\frac{1}{\alpha}}
			\\
			&=1-\sum_i\text{Tr}[(P_i\rho^\alpha P_i)^{\frac{1}{\alpha}}].
		\end{align*}
		This simplifies the nonlinear programming formulation. Hence, we have:
		\begin{align*}
			C_{0.5,1}(\rho,\mathbf{P})&=1-\sum_i\text{Tr}[(P_i\sqrt{\rho} P_i)^2]
			\\
			&=\sum_i\text{Tr}(P_i\rho)-\sum_i\text{Tr}[P_i\sqrt{\rho} P_i\cdot P_i\sqrt{\rho} P_i]
			\\
			&=\sum_i\left[ \text{Tr}(P_i\rho)- P_i\sqrt{\rho}P_i\sqrt{\rho} \right]
			\\
			&=\sum_i\left[ -\frac{1}{2}\text{Tr}\left(\sqrt{\rho},P_i\right)\right]
			\\
			&=\sum_iI_{WY}(\rho,P_i).
		\end{align*}
		Here, $I_{WY}(\rho,A)=-\frac{1}{2}\text{Tr}[\sqrt{\rho},A]^2$ represents the Wigner-Yanase skew information, and $[A,B]=AB-BA$ denotes the commutator. This leads to a block coherence measure defined by $I_{WY}$.
	\end{example}
	\par
	Next, we will further explore the intrinsic properties of this family of measures. These properties help us understand how the parameters $\alpha$ and $z$ influence the measure's value, as well as the measure's behavior under common quantum operations. As a two-parameter family of block coherence measures, the value of $C_{\alpha,z}$ depends on both parameters $\alpha$ and $z$, and changes when these parameters vary. The following conclusions can be drawn:
	\begin{proposition}
		For the parameters $\alpha$ and $z$ of the block coherence measure $C_{\alpha,z}$, the following holds:
		\par\noindent
		(1) When $\alpha$ is fixed, the value of $C_{\alpha,z}(\rho,\mathbf{P})$ increases as $z$ increases;
		\par\noindent
		(2) When $z=1$, if $\alpha_1\leqslant\alpha_2$, then $C_{\alpha_1,z}(\rho,\mathbf{P})\geqslant C_{\alpha_2,z}(\rho,\mathbf{P})$.
	\end{proposition}
	\begin{proof}
		\par\noindent\par\noindent
		(1) From Lemma 3.1, we know that when $z_1\leqslant z_2$, for $C_{\alpha,z_2}(\rho,\mathbf{P})$, there exists a free state $\gamma\in\mathcal{I}_{\text{B}}$ such that
		\begin{eqnarray*}
			C_{\alpha,z_2}(\rho,\mathbf{P})=1-\left[ \text{Tr}\left( \gamma^{\frac{1-\alpha}{2z_2}}\rho^{\frac{\alpha}{z_2}}\gamma^{\frac{1-\alpha}{2z_2}}\right)^{z_2}\right]^{\frac{1}{\alpha}}.
		\end{eqnarray*}
		Then we have:
		\begin{align*}
			\text{Tr}\left( \gamma^{\frac{1-\alpha}{2z_2}}\rho^{\frac{\alpha}{z_2}}\gamma^{\frac{1-\alpha}{2z_2}}\right)^{z_2}
			&=\text{Tr}\left( \gamma^{\frac{1-\alpha}{2z_2}}\rho^{\frac{\alpha}{z_2}}\gamma^{\frac{1-\alpha}{2z_2}}\right)^{\frac{z_2}{z_1}z_1}
			\\
			&\leqslant\text{Tr}\left( \gamma^{\frac{1-\alpha}{2z_1}}\rho^{\frac{\alpha}{z_1}}\gamma^{\frac{1-\alpha}{2z_1}}\right)^{z_1}
			\\
			&\leqslant\max_{\sigma\in\mathcal{I}_{\text{B}}}\text{Tr}\left( \sigma^{\frac{1-\alpha}{2z_1}}\rho^{\frac{\alpha}{z_1}}\sigma^{\frac{1-\alpha}{2z_1}}\right)^{z_1},
		\end{align*}
		which yields
		\begin{eqnarray*}
			C_{\alpha,z_1}(\rho,\mathbf{P})\leqslant C_{\alpha,z_2}(\rho,\mathbf{P}).
		\end{eqnarray*}
		\par\noindent
		(2) The simplified form via nonlinear programming has been given in Example 2:
		\begin{eqnarray*}
			C_{\alpha,1}(\rho,\mathbf{P})=1-\sum_i\text{Tr}[(P_i\rho^\alpha P_i)^{\frac{1}{\alpha}}].
		\end{eqnarray*}
		Then for each $i$, if $\alpha_1\leqslant\alpha_2$, we have:
		\begin{align*}
			\text{Tr}(P_i\rho^{\alpha_1}P_i)^{\frac{1}{\alpha_1}}&=\text{Tr}(P_i\rho^{\alpha_1}P_i)^{\frac{1}{\alpha_2}\frac{\alpha_2}{\alpha_1}}
			\\
			&\leqslant\text{Tr}(P_i^\frac{\alpha_2}{\alpha_1}\rho^{\alpha_2}P_i^\frac{\alpha_2}{\alpha_1})^{\frac{1}{\alpha_1}}
			\\
			&=\text{Tr}(P_i\rho^{\alpha_2}P_i)^{\frac{1}{\alpha_2}},
		\end{align*}
		Thus, $\sum_i\text{Tr}[(P_i\rho^\alpha_1 P_i)^{\frac{1}{\alpha_1}}]\leqslant\sum_i\text{Tr}[(P_i\rho^\alpha_2 P_i)^{\frac{1}{\alpha_2}}]$, which implies $C_{\alpha_1,z}(\rho,\mathbf{P})\geqslant C_{\alpha_2,z}(\rho,\mathbf{P})$.
	\end{proof}
	\par
	Beyond parameter dependence, we can explore its properties under broader operations. These properties help understand the measure's behavior and its relationship with other operations or systems. The following conclusions can be obtained:
	\begin{proposition}
		The block coherence measure $C_{\alpha,z}$ satisfies the following properties:
		\par\noindent
		(1) For a unitary matrix $U$ satisfying $U=\Delta(U)$, it holds that $C_{\alpha,z}(\rho,\mathbf{P})\geqslant C_{\alpha,z}(U\rho U^{\dagger},\mathbf{P})$;
		\par\noindent
		(2) For any density matrix $\tau\in H_\tau$ satisfying $\rho\notin H_\tau$, and a set of projectors $\mathbf{T}=\left\lbrace T_k\right\rbrace_k $ in $H_\tau$, it holds that $C_{\alpha,z}(\rho\otimes\tau,\mathbf{P}\otimes\mathbf{T})\geqslant C_{\alpha,z}(\rho,\mathbf{P})C_{\alpha,z}(\tau,\mathbf{T})$;
		\par\noindent
		(3) For any density matrix $\tau\in H_\tau$ satisfying $\rho\notin H_\tau$, and a set of projectors $\mathbf{T}=\left\lbrace T_k\right\rbrace_k $ in $H_\tau$, it holds that $C_{\alpha,z}(\rho\oplus\tau,\mathbf{P}\oplus\mathbf{T})\leqslant C_{\alpha,z}(\rho,\mathbf{P})+C_{\alpha,z}(\tau,\mathbf{T})$.
	\end{proposition}
	\begin{proof}
		\par\noindent\par\noindent
		(1) From literature \cite{audenaert2015alpha}, we know that $D_{\alpha,z}$ has unitary invariance: $D_{\alpha,z}(U\rho U^{\dagger}||U\sigma U^{\dagger})=D_{\alpha,z}(\rho||\sigma)$. This is equivalent to
		\begin{eqnarray*}
			\text{Tr}\left[(U\sigma U^{\dagger})^{\frac{1-\alpha}{2z}}(U\rho U^{\dagger})^{\frac{\alpha}{z}}(U\sigma U^{\dagger})^{\frac{1-\alpha}{2z}}\right] ^z=\text{Tr}(\sigma^{\frac{1-\alpha}{2z}}\rho^{\frac{\alpha}{z}}\sigma^{\frac{1-\alpha}{2z}})^z.
		\end{eqnarray*}
		For $C_{\alpha,z}(\rho,\mathbf{P})$, there exists $\gamma\in\mathcal{I}_{\text{B}}$ such that
		\begin{eqnarray*}
			C_{\alpha,z}(\rho,\mathbf{P})=1-\left\lbrace \text{Tr}(\gamma^{\frac{1-\alpha}{2z}}\rho^{\frac{\alpha}{z}}\gamma^{\frac{1-\alpha}{2z}})^z \right\rbrace^{\frac{1}{\alpha}},
		\end{eqnarray*}
		Noting that $U\gamma U^{\dagger}\in\mathcal{I}_{\text{B}}$, we have:
		\begin{align*}
			\text{Tr}(\gamma^{\frac{1-\alpha}{2z}}\rho^{\frac{\alpha}{z}}\gamma^{\frac{1-\alpha}{2z}})^z
			&=\text{Tr}\left[(U\gamma U^{\dagger})^{\frac{1-\alpha}{2z}}(U\rho U^{\dagger})^{\frac{\alpha}{z}}(U\gamma U^{\dagger})^{\frac{1-\alpha}{2z}}\right] ^z
			\\
			&\leqslant\max_{\sigma\in\mathcal{I}_{\text{B}}} \text{Tr}\left[ \sigma^{\frac{1-\alpha}{2z}}(U\rho U^{\dagger})^{\frac{\alpha}{z}}\sigma^{\frac{1-\alpha}{2z}}\right] ^z
		\end{align*}
		Consequently, $C_{\alpha,z}(\rho,\mathbf{P})\geqslant C_{\alpha,z}(U\rho U^{\dagger},\mathbf{P})$.
		\par\noindent
		(2) There exists $\gamma\in H_{\rho}$ such that $C_{\alpha,z}(\rho,\mathbf{P})=1-\left\lbrace \text{Tr}(\gamma^{\frac{1-\alpha}{2z}}\rho^{\frac{\alpha}{z}}\gamma^{\frac{1-\alpha}{2z}})^z \right\rbrace^{\frac{1}{\alpha}}$. There exists $\xi\in H_{\tau}$ such that $C_{\alpha,z}(\tau,\mathbf{T})=1-\left\lbrace \text{Tr}(\xi^{\frac{1-\alpha}{2z}}\tau^{\frac{\alpha}{z}}\xi^{\frac{1-\alpha}{2z}})^z \right\rbrace^{\frac{1}{\alpha}}$. Since
		\begin{align*}
			\gamma\otimes\xi&=\left(\sum_jP_j\gamma P_j \right)\otimes\left( \sum_kT_k\xi T_k\right)
			\\
			&=\sum_{j,k}(P_j\otimes T_k)(\gamma\otimes\xi)(P_j\otimes T_k),
		\end{align*}
		the set $\left\lbrace P_j\otimes T_k\right\rbrace_{j,k} $ can be considered as a set of projectors in the tensor system $H_\rho\otimes H_\tau$, and $\gamma\otimes\xi\in H_\rho\otimes H_\tau$ is a free state associated with these projectors. Thus, we have:
		\begin{align*}
			&\hspace{1.4em}\max_{\sigma\in\mathcal{I}^{H_\rho\otimes H_\tau}_{\text{B}}}\text{Tr}\left[\sigma^{\frac{1-\alpha}{2z}}(\rho\otimes\tau)^{\frac{\alpha}{z}}\sigma^{\frac{1-\alpha}{2z}} \right]^z
			\\
			&\geqslant\text{Tr}\left[(\gamma\otimes\xi)^{\frac{1-\alpha}{2z}}(\rho\otimes\tau)^{\frac{\alpha}{z}}(\gamma\otimes\xi)^{\frac{1-\alpha}{2z}} \right]^z
			\\
			&=\text{Tr}\left[(\gamma^{\frac{1-\alpha}{2z}}\otimes\xi^{\frac{1-\alpha}{2z}})(\rho^{\frac{\alpha}{z}}\otimes\tau^{\frac{\alpha}{z}})(\gamma^{\frac{1-\alpha}{2z}}\otimes\xi^{\frac{1-\alpha}{2z}}) \right]^z
			\\
			&=\text{Tr}\left[ (\gamma^{\frac{1-\alpha}{2z}}\rho^{\frac{\alpha}{z}}\gamma^{\frac{1-\alpha}{2z}})\otimes(\xi^{\frac{1-\alpha}{2z}}\tau^{\frac{\alpha}{z}}\xi^{\frac{1-\alpha}{2z}})\right]^z
			\\
			&=\text{Tr}(\gamma^{\frac{1-\alpha}{2z}}\rho^{\frac{\alpha}{z}}\gamma^{\frac{1-\alpha}{2z}})^z\text{Tr}(\xi^{\frac{1-\alpha}{2z}}\tau^{\frac{\alpha}{z}}\xi^{\frac{1-\alpha}{2z}})^z ,
		\end{align*}
		In the proof of non-negativity in Theorem 3.3, we have already obtained $\text{Tr}(\xi^{\frac{1-\alpha}{2z}}\tau^{\frac{\alpha}{z}}\xi^{\frac{1-\alpha}{2z}})^z\leqslant 1$. Since $\frac{1}{z}>1$, by Lemma 3.1 we know
		\begin{align*}
			\text{Tr}(\xi^{\frac{1-\alpha}{2z}}\tau^{\frac{\alpha}{z}}\xi^{\frac{1-\alpha}{2z}})^z&\geqslant \text{Tr}(\xi^{\frac{1-\alpha}{2}}\tau^{\alpha}\xi^{\frac{1-\alpha}{2}})^{\frac{1}{z}z}
			\\
			&=\text{Tr}(\xi^{\frac{1-\alpha}{2}}\tau^{\alpha}\xi^{\frac{1-\alpha}{2}})
			\\
			&=\text{Tr}(\tau^{\alpha}\xi^{1-\alpha}),
		\end{align*}
		Performing spectral decomposition on $\tau$ and $\xi$ yields $\tau=\sum_m\lambda_m\ket{\lambda_m}\bra{\lambda_m}$, $\xi=\sum_n\mu_n\ket{\mu_n}\bra{\mu_n}$, thus
		\begin{align*}
			\text{Tr}(\tau^{\alpha}\xi^{1-\alpha})&=\left(\sum_m\lambda_m^{\alpha}\ket{\lambda_m}\bra{\lambda_m} \right) \left(\sum_n\mu_n^{1-\alpha}\ket{\mu_n}\bra{\mu_n} \right)
			\\
			&=\sum_m\sum_n\lambda_m^{\alpha}\mu_n^{1-\alpha}\text{Tr}(\ket{\lambda_m}\bra{\lambda_m}\ket{\mu_n}\bra{\mu_n})
			\\
			&=\sum_m\sum_n\lambda_m^{\alpha}\mu_n^{1-\alpha}|\bra{\lambda_m}\ket{\mu_n}|^2
			\\
			&\geqslant0
		\end{align*}
		Therefore,
		\begin{align*}
			&\hspace{1.2em} C_{\alpha,z}(\rho\otimes\tau,\mathbf{P}\otimes\mathbf{T})
			\\
			&=1-\max_{\sigma\in\mathcal{I}^{H_\rho\otimes H_\tau}_{\text{B}}}\text{Tr}\left[\sigma^{\frac{1-\alpha}{2z}}(\rho\otimes\tau)^{\frac{\alpha}{z}}\sigma^{\frac{1-\alpha}{2z}} \right]^z
			\\
			&\geqslant1-\text{Tr}(\gamma^{\frac{1-\alpha}{2z}}\rho^{\frac{\alpha}{z}}\gamma^{\frac{1-\alpha}{2z}})^z\text{Tr}(\xi^{\frac{1-\alpha}{2z}}\tau^{\frac{\alpha}{z}}\xi^{\frac{1-\alpha}{2z}})^z
			\\
			&\geqslant[1-\text{Tr}(\gamma^{\frac{1-\alpha}{2z}}\rho^{\frac{\alpha}{z}}\gamma^{\frac{1-\alpha}{2z}})^z][1-\text{Tr}(\xi^{\frac{1-\alpha}{2z}}\tau^{\frac{\alpha}{z}}\xi^{\frac{1-\alpha}{2z}})^z]
			\\
			&=C_{\alpha,z}(\rho,\mathbf{P})C_{\alpha,z}(\tau,\mathbf{T})
		\end{align*}
		\par\noindent
		(3) There exists $\gamma\in H_{\rho}$ such that $C_{\alpha,z}(\rho,\mathbf{P})=1-\left\lbrace \text{Tr}(\gamma^{\frac{1-\alpha}{2z}}\rho^{\frac{\alpha}{z}}\gamma^{\frac{1-\alpha}{2z}})^z \right\rbrace^{\frac{1}{\alpha}}$. There exists $\xi\in H_{\tau}$ such that $C_{\alpha,z}(\tau,\mathbf{T})=1-\left\lbrace \text{Tr}(\xi^{\frac{1-\alpha}{2z}}\tau^{\frac{\alpha}{z}}\xi^{\frac{1-\alpha}{2z}})^z \right\rbrace^{\frac{1}{\alpha}}$. For notational convenience, we express $\mathbf{P}\oplus\mathbf{T}=\left\lbrace P_j\right\rbrace_j\cup\left\lbrace T_k\right\rbrace_k $ in another way, i.e., $\mathbf{P}\oplus\mathbf{T}:=\mathbf{M}=\left\lbrace M_l\right\rbrace_l $, such that
		\begin{eqnarray*}
			\gamma\oplus\xi=\left( \sum_jP_j\gamma P_j\right) \oplus\left( \sum_kT_k\xi T_k\right) =\sum_lM_l(\gamma\oplus\xi)M_l,
		\end{eqnarray*}
		Now, $\left\lbrace M_l\right\rbrace_l$ can be considered as a set of projectors in $\hat{H}$ generated by the basis vectors of $H_\rho$ and $H_\tau$, and $\gamma\oplus\xi\in \hat{H}$ is a free state associated with these projectors. Thus, we have:
		\begin{align*}
			&\hspace{1.4em}\max_{\sigma\in\mathcal{I}^{\hat{H}}_{\text{B}}}\text{Tr}\left[\sigma^{\frac{1-\alpha}{2z}}(\rho\oplus\tau)^{\frac{\alpha}{z}}\sigma^{\frac{1-\alpha}{2z}} \right]^z
			\\
			&\geqslant\text{Tr}\left[(\gamma\oplus\xi)^{\frac{1-\alpha}{2z}}(\rho\oplus\tau)^{\frac{\alpha}{z}}(\gamma\oplus\xi)^{\frac{1-\alpha}{2z}} \right]^z
			\\
			&=\text{Tr}\left[(\gamma^{\frac{1-\alpha}{2z}}\oplus\xi^{\frac{1-\alpha}{2z}})(\rho^{\frac{\alpha}{z}}\oplus\tau^{\frac{\alpha}{z}})(\gamma^{\frac{1-\alpha}{2z}}\oplus\xi^{\frac{1-\alpha}{2z}}) \right]^z
			\\
			&=\text{Tr}\left[ (\gamma^{\frac{1-\alpha}{2z}}\rho^{\frac{\alpha}{z}}\gamma^{\frac{1-\alpha}{2z}})\oplus(\xi^{\frac{1-\alpha}{2z}}\tau^{\frac{\alpha}{z}}\xi^{\frac{1-\alpha}{2z}})\right]^z
			\\
			&=\text{Tr}(\gamma^{\frac{1-\alpha}{2z}}\rho^{\frac{\alpha}{z}}\gamma^{\frac{1-\alpha}{2z}})^z+\text{Tr}(\xi^{\frac{1-\alpha}{2z}}\tau^{\frac{\alpha}{z}}\xi^{\frac{1-\alpha}{2z}})^z,
		\end{align*}
		Consequently, we have
		\begin{align*}
			C_{\alpha,z}(\rho\oplus\tau,\mathbf{P}\oplus\mathbf{T})
			&\leqslant1+C_{\alpha,z}(\rho\oplus\tau,\mathbf{P}\oplus\mathbf{T})
			\\
			&=2-\max_{\sigma\in\mathcal{I}^{\hat{H}}_{\text{B}}}\text{Tr}\left[\sigma^{\frac{1-\alpha}{2z}}(\rho\oplus\tau)^{\frac{\alpha}{z}}\sigma^{\frac{1-\alpha}{2z}} \right]^z
			\\
			&\leqslant1-\text{Tr}(\gamma^{\frac{1-\alpha}{2z}}\rho^{\frac{\alpha}{z}}\gamma^{\frac{1-\alpha}{2z}})^z+1-\text{Tr}(\xi^{\frac{1-\alpha}{2z}}\tau^{\frac{\alpha}{z}}\xi^{\frac{1-\alpha}{2z}})^z
			\\
			&=
			C_{\alpha,z}(\rho,\mathbf{P})+C_{\alpha,z}(\tau,\mathbf{T}).
		\end{align*}
	\end{proof}
	\par
	As the final part of this section, we note that the preceding discussion on $C_{\alpha,z}(\rho,\mathbf{P})$ is conducted within the general framework of block coherence resource theory. An important research direction is applying the general theory to concrete quantum biological scenarios. To measure S-T coherence, take the projector set as $\mathbf{Q}=\left\lbrace Q_S,Q_T\right\rbrace$. Then, a block coherence measure $C(\rho,\mathbf{P})$ induces a corresponding S-T coherence measure $C(\rho,\mathbf{Q})$. For $C_{\alpha,z}$, the following conclusion holds:
	\begin{proposition}
		The function $C_{\alpha,z}$ defined via the $\alpha$-$z$ Rényi relative entropy is a valid S-T coherence measure, defined as:
		\begin{eqnarray*}
			C_{\alpha,z}(\rho,\mathbf{Q})=1-\max_{\sigma\in\mathcal{I}_{\text{ST}}}\left\lbrace \text{Tr}(\sigma^{\frac{1-\alpha}{2z}}\rho^{\frac{\alpha}{z}}\sigma^{\frac{1-\alpha}{2z}})^z \right\rbrace^{\frac{1}{\alpha}},
		\end{eqnarray*}
		where $\gamma\in\mathcal{I}_{\text{ST}}$ is the free state in S-T coherence resource theory, i.e., $\gamma=\Delta(\gamma)=Q_S\gamma Q_S+Q_T\gamma Q_T$, with $\alpha\in(0,1)$ and $\text{max}\left\lbrace \alpha,1-\alpha\right\rbrace \leqslant z$.
	\end{proposition}
	
	\subsection{Block coherence measures induced by Tsallis relative operator entropy}
	The Tsallis relative operator entropy, as a generalization of relative entropy, has garnered significant attention in academia recently. For density matrices $A$ and $B$, the Tsallis relative operator entropy is defined in literature \cite{bikchentaev2024trace} as:
	\begin{align*}
		NT_{1-\beta}(A||B)=\frac{1}{1-\beta}[1-\text{Tr}(A\sharp_{1-\beta}B)],
	\end{align*}
	where $A\sharp_{1-\beta}B=A^{\frac{1}{2}}(A^{-\frac{1}{2}}BA^{-\frac{1}{2}})^{1-\beta}A^{\frac{1}{2}}$.
	\par
	To construct a family of block coherence measures based on the Tsallis relative operator entropy, we present the following lemma \cite{bikchentaev2024trace}:
	\begin{lemma}
		For any $\beta\in(0,1)$, the following holds:
		\begin{align*}
			\text{Tr}(A\sharp_{1-\beta}B)\leqslant\text{Tr}(A^{\beta}B^{1-\beta}).
		\end{align*}
	\end{lemma}
	\par
	Utilizing the above lemma, we can now define a family of block coherence measures based on the Tsallis relative operator entropy and verify that it satisfies the required conditions for a measure. We have the following theorem:
	\begin{theorem}
		The function $C^N_{\beta}$ defined via the Tsallis relative operator entropy is a valid block coherence measure, defined as:
		\begin{eqnarray*}
			C^N_{\beta}(\rho,\mathbf{P})=\frac{1}{1-\beta}\left\lbrace 1-\max_{\sigma\in\mathcal{I}_{\text{B}}}\big[ \text{Tr}(\rho\sharp_{1-\beta}\sigma)\big]^{\frac{1}{\beta}}\right\rbrace ,
		\end{eqnarray*}
		where $\beta\in(0,1)$.
	\end{theorem}
	\begin{proof}
		\par\noindent\par
		We complete this proof via the alternative framework.
		\par\noindent
		(B1) Non-negativity: Using Lemma 3.1, we can obtain:
		\begin{align*}
			\text{Tr}(\sigma^{\frac{1-\alpha}{2z}}\rho^{\frac{\alpha}{z}}\sigma^{\frac{1-\alpha}{2z}})^z&\geqslant \text{Tr}(\sigma^{\frac{1-\alpha}{2}}\rho^{\alpha}\sigma^{\frac{1-\alpha}{2}})^{\frac{1}{z}z}
			\\
			&=\text{Tr}(\sigma^{\frac{1-\alpha}{2}}\rho^{\alpha}\sigma^{\frac{1-\alpha}{2}})
			\\
			&=\text{Tr}(\rho^{\alpha}\sigma^{1-\alpha}).
		\end{align*}
		The proof of Theorem 3.3 has already established $\text{Tr}(\sigma^{\frac{1-\alpha}{2z}}\rho^{\frac{\alpha}{z}}\sigma^{\frac{1-\alpha}{2z}})^z\leqslant 1$. Combining this with Lemma 3.2 yields:
		\begin{align*}
			\text{Tr}(\rho\sharp_{1-\beta}\sigma)\leqslant\text{Tr}(\rho^{\beta}\sigma^{1-\beta})\leqslant\text{Tr}(\sigma^{\frac{1-\beta}{2z}}\rho^{\frac{\beta}{z}}\sigma^{\frac{1-\beta}{2z}})^z\leqslant 1.
		\end{align*}
		Thus, $\max_{\sigma\in\mathcal{I}_{\text{B}}}\left\lbrace \text{Tr}(\rho\sharp_{1-\beta}\sigma) \right\rbrace^{\frac{1}{\beta}}\leqslant 1$, which implies $C^N_{\beta}(\rho,\mathbf{P})\geqslant 0$.
		\par
		When $\rho\in\mathcal{I}_{\text{B}}$, we have:
		\begin{eqnarray*}
			1=\text{Tr}(\rho)=\text{Tr}(\rho\sharp_{1-\beta}\rho)\leqslant\max_{\sigma\in\mathcal{I}_{\text{B}}}\text{Tr}(\rho\sharp_{1-\beta}\sigma),
		\end{eqnarray*}
		hence $\max_{\sigma\in\mathcal{I}_{\text{B}}}\left\lbrace \text{Tr}(\rho\sharp_{1-\beta}\sigma) \right\rbrace^{\frac{1}{\beta}}=1$, and consequently $C^N_{\beta}(\rho,\mathbf{P})=0$.
		\par
		When $C^N_{\beta}(\rho,\mathbf{P})=0$, there exists $\gamma\in\mathcal{I}_{\text{B}}$ such that $\text{Tr}(\rho\sharp_{1-\beta}\gamma)=1$, leading to $\text{Tr}(\rho^{\beta}\gamma^{1-\beta})=1$. According to the non-negativity of the Tsallis relative entropy \cite{abe2003monotonic}, it follows that $\rho=\gamma\in\mathcal{I}_{\text{B}}$. This completes the proof of non-negativity.
		\par\noindent
		(B2) Monotonicity: Literature \cite{bikchentaev2024trace} provides the monotonicity of $NT_{1-\beta}(A||B)$, i.e., $NT_{1-\beta}[\Phi(\rho)||\Phi(\sigma)]\leqslant NT_{1-\beta}(\rho||\sigma)$. This is equivalent to $\text{Tr}[\Phi(\rho)\sharp_{1-\beta}\Phi(\sigma)]\geqslant\text{Tr}(\rho\sharp_{1-\beta}\sigma)$. Therefore, there exists $\gamma\in\mathcal{I}_{\text{B}}$ such that:
		\begin{align*}
			C^N_{\beta}(\rho,\mathbf{P})
			&=
			1-\left\lbrace \text{Tr}(\text{Tr}(\rho\sharp_{1-\beta}\gamma) \right\rbrace^{\frac{1}{\beta}}
			\\
			&\geqslant
			1-\left\lbrace \text{Tr}[\Phi(\rho)\sharp_{1-\beta}\Phi(\gamma)] \right\rbrace^{\frac{1}{\beta}}
			\\
			&\geqslant
			1-\max_{\sigma\in\mathcal{I}_{\text{B}}}\left\lbrace \text{Tr}[\Phi(\rho)\sharp_{1-\beta}\sigma] \right\rbrace^{\frac{1}{\alpha}}
			\\
			&=C^N_{\beta}(\Phi(\rho),\mathbf{P}).
		\end{align*}
		This proves monotonicity.
		\par\noindent
		(B5) Block Additivity: In the requirement of block additivity, let $\rho=p_1\rho_1\oplus p_2\rho_2$. For any $\sigma\in\mathcal{I}_{\text{B}}$, there exists a decomposition such that $\sigma=q_1\sigma_1\oplus q_2\sigma_2$, where the coefficients $q_1>0$ and $q_2>0$ are undetermined constants satisfying $q_1+q_2=1$, and $\dim\sigma_j=\dim\rho_j$. Therefore, we know:
		\begin{align*}
			&\hspace{1.1em}\max_{\sigma}\text{Tr}(\rho\sharp_{1-\beta}\sigma)
			\\
			&=
			\max_{q_1, q_2}\left\lbrace q_1^{1-\beta}p_1^{\beta}\max_{\sigma_1}\text{Tr}(\rho_1\sharp_{1-\beta}\sigma_1)+q_2^{1-\beta}p_2^{\beta}\max_{\sigma_2}\text{Tr}(\rho_2\sharp_{1-\beta}\sigma_2)\right\rbrace
			\\
			&=
			\max_{q_1, q_2}(q_1^{1-\beta}p_1^{\beta}t_1+q_2^{1-\beta}p_2^{\beta}t_2),
		\end{align*}
		where $t_i=\max_{\sigma_i}\text{Tr}(\rho_i\sharp_{1-\beta}\sigma_i)$.
		\par
		Next, using the H{\"o}lder inequality, we have:
		\begin{align*}
			q_1^{1-\beta}p_1^{\beta}t_1+q_2^{1-\beta}p_2^{\beta}t_2&\leqslant \left[ \sum_{i}(q_i^{1-\beta})^{\frac{1}{1-\beta}}\right] ^{1-\beta}\left[ \sum_{i}(p_i^{\beta}t_i)^{\frac{1}{\beta}}\right] ^{\beta}
			\\
			&=\left(\sum_{i}q_i\right) ^{1-\beta}\left[ \sum_{i}(p_i^{\beta}t_i)^{\frac{1}{\beta}}\right] ^{\beta}
			\\
			&=\left[ \sum_{i}(p_i^{\beta}t_i)^{\frac{1}{\beta}}\right] ^{\beta},
		\end{align*}
		The inequality becomes an equality if and only if $q_i=cp_it_i^{\frac{1}{\beta}}$, where $c=(p_1t_1^{\frac{1}{\beta}}+p_2t_2^{\frac{1}{\beta}})^{-1}$. Thus, the originally undetermined $q_1$ and $q_2$ can be determined by $p_j$ and $t_j$. We have:
		\begin{eqnarray*}
			\max_{\sigma}\left\lbrace \text{Tr}(\rho\sharp_{1-\beta}\sigma) \right\rbrace^{\frac{1}{\beta}}=\sum_{i}(p_i^{\beta}t_i)^{\frac{1}{\beta}}=\sum_{i}p_i\max_{\sigma_i}\left\lbrace \text{Tr}(\rho_i\sharp_{1-\beta}\sigma_i) \right\rbrace^{\frac{1}{\beta}}.
		\end{eqnarray*}
		This implies $C^N_{\beta}(\rho,\mathbf{P})=p_1C^N_{\beta}(\rho_1,\mathbf{P})+p_2C^N_{\beta}(\rho_2,\mathbf{P})$, i.e., additivity holds.
		\par
		So far, we have verified conditions (B1), (B2), and (B5) to confirm that $C^N_{\beta}(\rho,\mathbf{P})$ is a valid block coherence measure.
	\end{proof}
	\section{Relations among Different Block Coherence Measures}
	\subsection{Review of Existing Measures}
	To compare the ordering and numerical relationships among different block coherence measures, we first conduct a systematic review and summary of the measures that have been proposed in the literature within the framework of block coherence resource theory and its important specific instance, S-T coherence resource theory.
	\par
	In the general framework of block coherence resource theory, the following measures have been extensively studied:
	\par
	Literature \cite{bischof2021quantifying} gives the definition of the robustness of block coherence measure as:
	\begin{align}\label{eq3}
		C_{rob}(\rho,\mathbf{P})&=\min_{\tau\in\mathcal{D}(H)}\left\lbrace s\geqslant 0|\frac{\rho+s\tau}{1+s}\in\mathcal{I}_{\text{B}} \right\rbrace \nonumber
		\\
		&=\min_{\delta\in\mathcal{I}_{\text{B}}}\left\lbrace s\geqslant 0|\rho\leqslant(1+s)\Delta(\sigma)\right\rbrace
		\\
		&=\max_{X\geqslant 0}\left\lbrace \text{Tr}(X\rho)-1|\Delta(X)=I\right\rbrace.
	\end{align}
	\par
	Literature \cite{xu2020general} provides the block coherence measure based on the trace norm as:
	\begin{eqnarray*}
		C_{\text{Tr}}(\rho,\mathbf{P})=\min_{\lambda>0,\sigma\in\mathcal{I}_{\text{B}}}\norm{\rho-\lambda\sigma}_\text{Tr},
	\end{eqnarray*}
	where the trace norm is defined as $\norm{M}_{\text{Tr}}=\text{Tr}\sqrt{M^{\dagger}M}$.
	\par
	Meanwhile, literature \cite{xu2020general} also provides the block coherence measure based on the Tsallis relative entropy as:
	\begin{eqnarray*}
		C^T_{\alpha}(\rho,\mathbf{P})=\frac{1}{1-\alpha}\left\lbrace 1-\max_{\sigma\in\mathcal{I}_{\text{B}}}\big[\text{Tr}(\rho^{\alpha}\sigma^{1-\alpha})\big]^{\frac{1}{\alpha}}\right\rbrace =\frac{1}{1-\alpha}\left\lbrace 1-\sum_i\text{Tr}\big[(P_i\rho^{\alpha}P_i)^{\frac{1}{\alpha}}\big]\right\rbrace.
	\end{eqnarray*}
	\par
	Literature \cite{lei2021povm} gives the geometric block coherence measure, defined as:
	\begin{eqnarray*}
		C_{geo}(\rho,\mathbf{P})=\min_{\sigma\in\mathcal{D}(H)}D_{geo}(\rho,\Delta(\sigma))
	\end{eqnarray*}
	where $D_{geo}(\rho,\sigma)=1-F^2(\rho,\sigma)$, and $F(\rho,\sigma)=\text{Tr}\sqrt{\sqrt{\rho}\sigma\sqrt{\rho}}$ is the fidelity.
	\par
	Literature \cite{fu2022block} gives the max-relative entropy block coherence measure, defined as:
	\begin{eqnarray*}
		C_{max}(\rho,\mathbf{P})=\min_{\sigma\in\mathcal{I}_{\text{B}}}D_{max}(\rho||\sigma)=\min_{\sigma\in\mathcal{I}_{\text{B}}}\log\min_{\lambda}\left\lbrace \lambda\geqslant 1|\rho\leqslant\lambda\sigma\right\rbrace
	\end{eqnarray*}
	\par
	S-T coherence resource theory is a typical and important specific instance of block coherence resource theory, with the projector set $\mathbf{Q}=\left\lbrace Q_S,Q_T \right\rbrace $. Therefore, any block coherence measure can be specialized to yield a corresponding S-T coherence measure by setting $\mathbf{P}=\mathbf{Q}$.
	\par
	In 2014, Kominis proposed \cite{kritsotakis2014retrodictive} an S-T coherence function $C_{l_1}$ based on the $l_1$ norm, but later work \cite{kominis2020quantum} pointed out a deficiency of $C_{l_1}$: it yields higher coherence for states possessing a larger triplet component. Kominis also provided a more appropriate S-T coherence measure $C_r$ based on relative entropy, defined as:
	\begin{eqnarray*}
		C_r(\rho,\mathbf{Q})=\min_{\sigma\in\mathcal{I}_{\text{ST}}}S(\rho||\sigma)=S(\rho_{SS}+\rho_{TT})-S(\rho),
	\end{eqnarray*}
	where $S(A||B)=\text{Tr}[A\log(A)]-\text{Tr}[A\log(B)]$ is the quantum relative entropy, and $S(A)=-\text{Tr}[A\log(A)]$ is the von Neumann entropy.
	\par
	Considering the importance of an S-T coherence measure based on the $l_1$ norm, we can derive a more suitable improved version with the help of literature \cite{xu2020general}:
	\begin{eqnarray*}
		\widetilde{C}_{l_1}(\rho,\mathbf{Q})=\norm{\rho_{ST}}_{\text{Tr}}+\norm{\rho_{TS}}_{\text{Tr}}=2\norm{\rho_{ST}}_{\text{Tr}}.
	\end{eqnarray*}
	\par
	Literature \cite{kominis2025physiological} gives the S-T coherence measure based on the Wigner-Yanase skew information, defined as:
	\begin{eqnarray*}
		C_{WY}(\rho,\mathbf{Q})=I_{WY}(\rho,Q_S)+I_{WY}(\rho,Q_T),
	\end{eqnarray*}
	According to Example 2 in Section 3.1, this is precisely the manifestation of the block coherence measure $C_{0.5,1}(\rho,\mathbf{P})$ when the projector set is taken as $\mathbf{Q}$.
	\par
	Next, we will investigate the relationships among these measures.
	\subsection{Ordering Relations among Measures}
	In quantum resource theory, specific measures are needed to quantify the amount of resource a quantum state possesses. Studying ordering relations aims to clarify under what conditions different measures yield consistent conclusions. This has been discussed in entanglement resource theory \cite{miranowicz2004ordering}, coherence resource theory \cite{liu2016ordering}, and imaginarity resource theory \cite{chen2023measures}. In the context of block coherence resource theory, we can similarly define ordering relations:
	\begin{definition}
		Given two block coherence measures $C_1$ and $C_2$, if for any two quantum states $\rho,\sigma\in A$ in the set $A$, under the same projector set $\mathbf{P}$, the following relation always holds:
		\begin{eqnarray*}
			C_1(\rho,\mathbf{P})\leqslant C_1(\sigma,\mathbf{P})
			\Leftrightarrow
			C_2(\rho,\mathbf{P})\leqslant C_2(\sigma,\mathbf{P}),
		\end{eqnarray*}
		then the block coherence measures $C_1$ and $C_2$ are said to have the same ordering relation on the set $A$.
	\end{definition}
	\par
	From the proof of Proposition 3.1(2), we know that for any state, the following holds:
	\begin{eqnarray*}
		C_{\alpha,1}(\rho,\mathbf{Q})=1-\text{Tr}\left[ (Q_S\rho^{\alpha}Q_S)^{\frac{1}{\alpha}}+(Q_T\rho^{\alpha}Q_T)^{\frac{1}{\alpha}}\right].
	\end{eqnarray*}
	This analytical form can help us better study the ordering relation of $C_{\alpha,1}$ with respect to the parameter $\alpha$. We have:
	\begin{proposition}
		When $\rho$ is a pure state, the S-T coherence measure $C_{\alpha,1}$ has the same ordering relation for different $\alpha$; when $\rho$ is an arbitrary state, $C_{\alpha,1}$ does not have the same ordering relation for different $\alpha$.
	\end{proposition}
	\begin{proof}
		\par\noindent\par
		Substituting the pure state $\ket{\psi}=\alpha_S\ket{S}+\sum_{j=-1}^{1}\beta_j\ket{T_j}$ into the expression for $C_{\alpha,1}$, direct calculation yields:
		\begin{eqnarray*}
			C_{\alpha,1}(\rho,\mathbf{Q})=1-\alpha_S^{2/\alpha}-(1-\alpha_S^2)^{1/\alpha}.
		\end{eqnarray*}
		Define the function:
		\begin{eqnarray*}
			f_R(x,\alpha)=1-x^{\frac{1}{\alpha}}-(1-x)^{\frac{1}{\alpha}}.
		\end{eqnarray*}
		Then for any $\alpha$, $f_R(x,\alpha)=f_R(1-x,\alpha)$ holds, and it is monotonically increasing in $x$ on $(0,\frac{1}{2})$. This indicates that for pure states, $C_{\alpha,1}$ certainly has the same ordering relation for different $\alpha$.
		\par
		When $\rho$ is an arbitrary state, numerical simulations show otherwise. Taking $\rho_1$ and $\rho_2$ respectively as:
		\begin{eqnarray*}
			\begin{pmatrix}
				0.3476 & -0.1783 + 0.1361i & -0.0711 - 0.0503i & -0.1975 + 0.0413i \\
				-0.1783 - 0.1361i & 0.3182 & -0.1154 + 0.0703i & 0.0582 - 0.0169i \\
				-0.0711 + 0.0503i & -0.1154 - 0.0703i & 0.1498 & 0.0827 + 0.0416i \\
				-0.1975 - 0.0413i & 0.0582 + 0.0169i & 0.0827 - 0.0416i & 0.1844
			\end{pmatrix}
			\\
			\begin{pmatrix}
				0.2912 & -0.0327 - 0.0229i & -0.0670 - 0.2773i & 0.0165 + 0.1400i \\
				-0.0327 + 0.0229i & 0.0594 & 0.0133 + 0.0331i & 0.0362 - 0.0220i \\
				-0.0670 + 0.2773i & 0.0133 - 0.0331i & 0.3836 & -0.0696 - 0.0662i \\
				0.0165 - 0.1400i & 0.0362 + 0.0220i & -0.0696 + 0.0662i & 0.2658
			\end{pmatrix}
		\end{eqnarray*}
		The calculation results from $\alpha=0.1$ to $\alpha=0.9$ are as follows:
		\begin{align*}
			C_{0.1,1}(\rho_1)\approx0.2362&>0.2115\approx C_{0.1,1}(\rho_2);
			\\
			C_{0.2,1}(\rho_1)\approx0.1969&>0.1869\approx C_{0.2,1}(\rho_2);
			\\
			C_{0.3,1}(\rho_1)\approx0.1636&>0.1625\approx C_{0.3,1}(\rho_2);
			\\
			C_{0.4,1}(\rho_1)\approx0.1343&<0.1384\approx C_{0.4,1}(\rho_2);
			\\
			C_{0.5,1}(\rho_1)\approx0.1087&<0.1144\approx C_{0.5,1}(\rho_2);
			\\
			C_{0.6,1}(\rho_1)\approx0.0847&<0.0908\approx C_{0.6,1}(\rho_2);
			\\
			C_{0.7,1}(\rho_1)\approx0.0621&<0.0675\approx C_{0.7,1}(\rho_2);
			\\
			C_{0.8,1}(\rho_1)\approx0.0406&<0.0446\approx C_{0.8,1}(\rho_2);
			\\
			C_{0.9,1}(\rho_1)\approx0.0198&<0.0221\approx C_{0.9,1}(\rho_2);
		\end{align*}
		Therefore, for arbitrary states, $C_{\alpha,1}$ does not have the same ordering relation for different $\alpha$.
	\end{proof}
	\begin{remark}
		In the definition of the block coherence measure $C_{\alpha,z}(\rho,\mathbf{P})$, we stipulate $\alpha\neq 1$ because the result for $\alpha=1$ is trivial, i.e.:
		\begin{eqnarray*}
			C_{1,z}(\rho,\mathbf{P})=1-\max_{\sigma\in\mathcal{I}_{\text{B}}}\left\lbrace \text{Tr}(\sigma^{\frac{0}{2z}}\rho^{\frac{1}{z}}\sigma^{\frac{0}{2z}})^z \right\rbrace=1-\text{Tr}(\rho)\equiv 0.
		\end{eqnarray*}
		If we define the difference function $DIS(\alpha)=C_{\alpha,1}(\rho_1,\mathbf{Q})-C_{\alpha,1}(\rho_2,\mathbf{Q})$, where $\alpha\in(0,1)$,
		then $\lim_{\alpha\to 1^-}DIS(\alpha)\to 0^-$ is reasonable.
		\par
		In the proof of Proposition 4.1, we constructed a specific pair of quantum states $\rho_1$ and $\rho_2$. Through numerical experiments, inserting $1e6$ points between $\alpha\in(\frac{1e-6}{2},1-\frac{1e-6}{2})$ yields a unique zero $\alpha^*$ for $D(\alpha)$ around $0.3171$, with $DIS(\frac{1e-6}{2})\approx 0.4545$, $DIS(1-\frac{1e-6}{2})\approx -1.2363e-7$. We can then say that the difference function $DIS(\alpha)=C_{\alpha,1}(\rho_1,\mathbf{Q})-C_{\alpha,1}(\rho_2,\mathbf{Q})$ has a unique zero on $\alpha\in(0,1)$, meaning the ordering relation reverses exactly once:
		\begin{align*}
			C_{\alpha,1}(\rho_1,\mathbf{Q})&>C_{\alpha,1}(\rho_2,\mathbf{Q}),\alpha<\alpha^*,
			\\
			C_{\alpha,1}(\rho_1,\mathbf{Q})&<C_{\alpha,1}(\rho_2,\mathbf{Q}),\alpha>\alpha^*.
		\end{align*}
		This seems to suggest that in this specific case, the ordering relation reverses only once as $\alpha$ increases. This raises a question: for any two given mixed states, does the ordering relation reverse at most once, i.e., does $D(\alpha)$ have at most one zero?
		\par
		However, further systematic numerical exploration indicates that this conjecture does not hold. Figure 4-1 shows that for a fixed specific pair of quantum states $\rho_1$ and $\rho_2$, $D(\alpha)$ can have two zeros ($\alpha_{z_1}$ and $\alpha_{z_2}$), meaning that as $\alpha$ changes, the size relationship between $C_{\alpha,1}(\rho_1,\mathbf{Q})$ and $C_{\alpha,1}(\rho_2,\mathbf{Q})$ can alternate multiple times (e.g., the sign alternation from $\alpha_1$ to $\alpha_2$ to $\alpha_3$ in the figure). Moreover, there exist many pairs of quantum states for which $D(\alpha)$ has more zeros (more than two), so the conjecture does not hold.
		\begin{figure}[H]
			\subfigure
			{
				\begin{minipage}[b]{0.58\linewidth}
					\includegraphics[height=0.17\textheight]{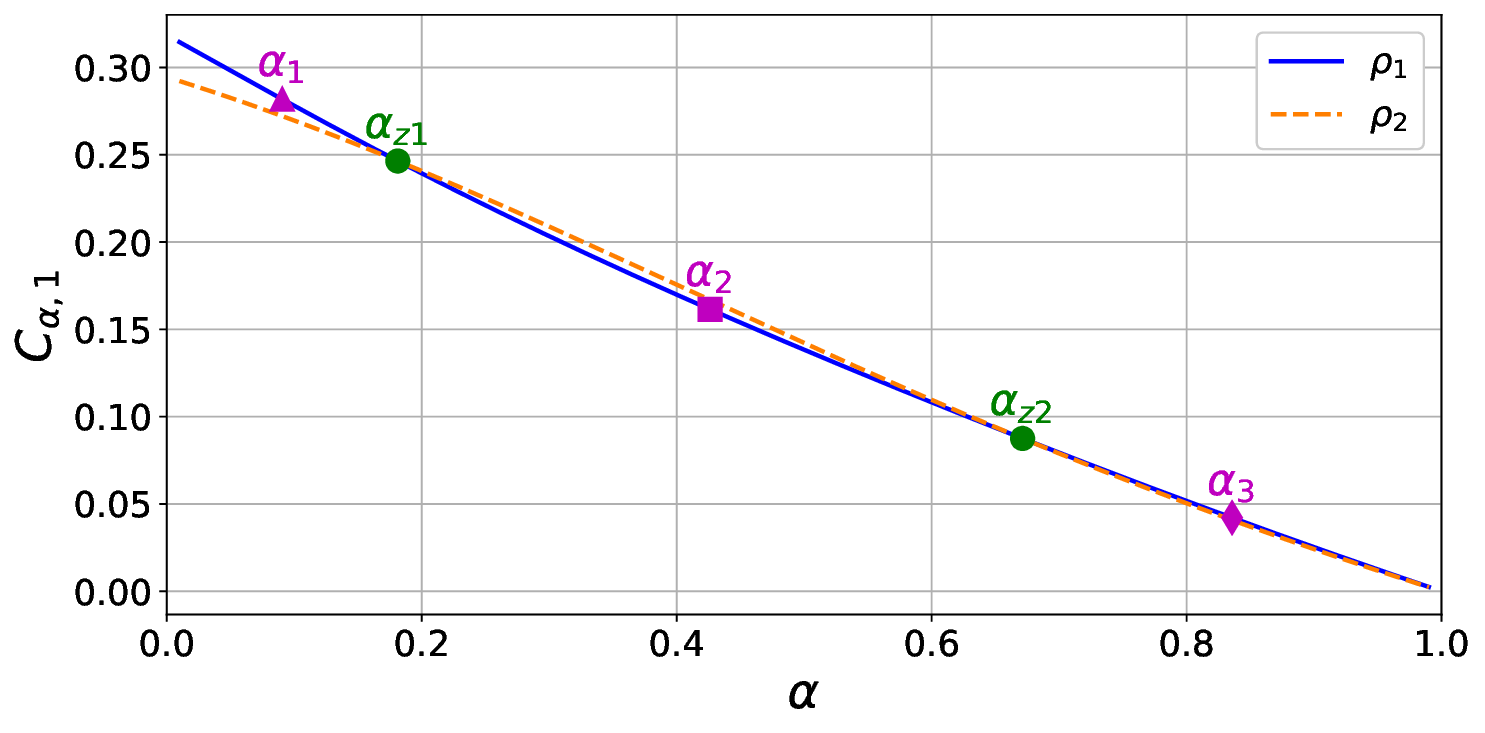}
				\end{minipage}
			}
			\subfigure
			{
				\begin{minipage}[b]{0.3\linewidth}
					\includegraphics[height=0.17\textheight]{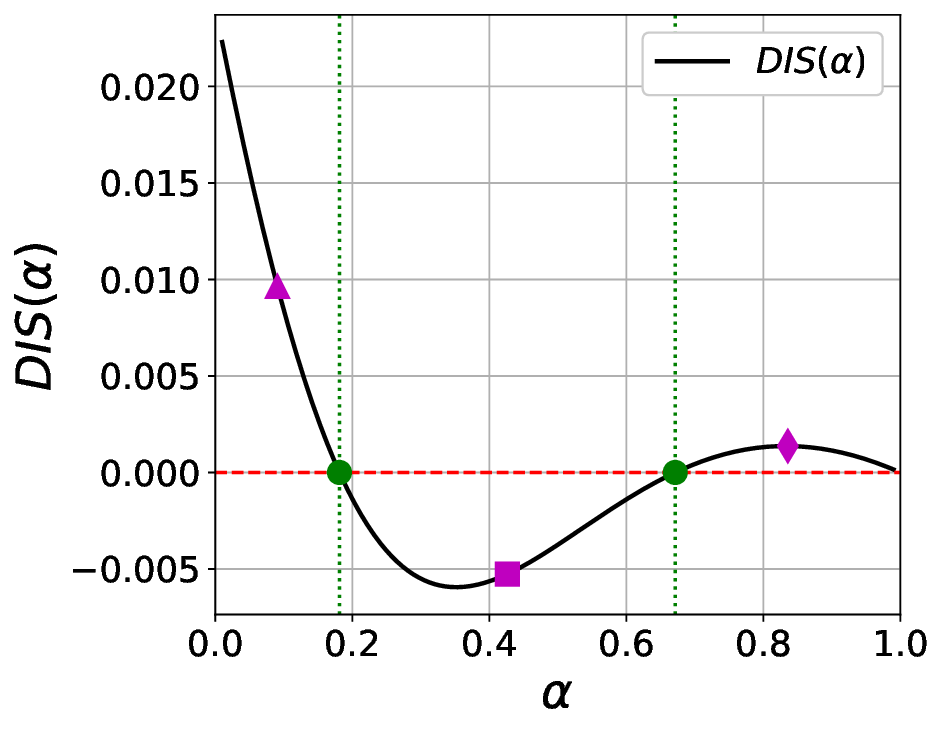}
				\end{minipage}
			}
			\caption{The block coherence measure $C_{\alpha,z}$ and the difference function $DIS(\alpha)$.}
		\end{figure}
		\par
		Although for some state pairs, $D(\alpha)$ may exhibit single-zero behavior, this is not a universal rule. Therefore, the conclusion in Proposition 4.1 that $C_{\alpha,1}$ does not have the same ordering relation for arbitrary states is rigorous and general, encompassing various possibilities ranging from constant ordering, single reversal, to multiple reversals.
	\end{remark}
	\par
	We can also prove the ordering among existing S-T coherence measures. Considering that $C_{WY}$ can be written as $C_{0.5,1}$, we only consider the following relationship:
	\begin{proposition}
		When $\rho$ is a pure state, the S-T coherence measures $C_{\alpha,1}$, $\widetilde{C}_{l_1}$, and $C_r$ have the same ordering; when $\rho$ is an arbitrary state, the three do not have the same ordering relation.
	\end{proposition}
	\begin{proof}
		\par\noindent\par
		Substituting the pure state $\ket{\psi}=\alpha_S\ket{S}+\sum_{j=-1}^{1}\beta_j\ket{T_j}$ into $\widetilde{C}_{l_1}$ and $C_r$, we obtain:
		\begin{align*} \widetilde{C}_{l_1}(\ket{\psi},\mathbf{Q})&=2\sqrt{\alpha_S^2(1-\alpha_S^2)},
			\\
			C_r(\ket{\psi},\mathbf{Q})&=-\alpha_S^2\log(\alpha_S^2)-(1-\alpha_S^2)\log(1-\alpha_S^2).
		\end{align*}
		Define the functions:
		\begin{align*}
			f_{l_1}(x)&=2\sqrt{x(1-x)},
			\\
			f_r(x)&=x\log(x)-(1-x)\log(1-x).
		\end{align*}
		Thus, analyzing $C_{\alpha,1}$, $\widetilde{C}_{l_1}$, and $C_r$ can be turned into directly analyzing the functions $f_R(x,\alpha)$, $f_{l_1}(x)$, and $f_r(x)$. Since $f_{l_1}(x)=f_{l_1}(1-x)$ and $f_r(x)=f_r(1-x)$, and all three are monotonically increasing in $x$ on $(0,\frac{1}{2})$, it can be deduced that the three measures have the same ordering for pure states.
		\par
		However, for arbitrary states, counterexamples exist where the ordering does not hold. For example, taking $\rho_1$ and $\rho_2$ respectively as
		\begin{eqnarray*}
			\begin{pmatrix}
				0.2564 & 0.2244+0.0068i & 0.0911-0.0216i & 0.1709 \\
				0.2244-0.0068i & 0.2317 & 0.1516+0.0477i & 0.1882-0.0273i \\
				0.0911+0.0216i & 0.1516-0.0477i & 0.3372 & 0.0963-0.1143i \\
				0.1709 & 0.1882+0.0273i & 0.0963+0.1143i & 0.1747
			\end{pmatrix},
			\\
			\begin{pmatrix}
				0.2871 & 0.2749-0.0193i & 0.1011+0.0191i & 0.1348+0.0513i \\
				0.2749+0.0193i & 0.4127 & 0.0708+0.0021i & 0.1415+0.0318i \\
				0.1011-0.0191i & 0.0708-0.0021i & 0.1398 & 0.0871+0.0375i \\
				0.1348-0.0513i & 0.1415-0.0318i & 0.0871-0.0375i & 0.1604
			\end{pmatrix}.
		\end{eqnarray*}
		Calculations yield:
		\begin{align*}
			0.299\approx\widetilde{C}_{l_1}(\rho_1,\mathbf{Q})&<\widetilde{C}_{l_1}(\rho_2,\mathbf{Q})\approx0.3793
			\\
			0.1906\approx C_r(\rho_1,\mathbf{Q})&>C_r(\rho_2,\mathbf{Q})\approx 0.1118
		\end{align*}
		Thus, they have different orderings; the other two relations are similar.
	\end{proof}
	\par
	Propositions 4.1 and 4.2 indicate that different block coherence measures (or different parameters of the same measure family) may induce different orderings on the set of quantum states. This reminds us that when comparing the coherence magnitude of different states, the specific measure employed must be clearly specified.
	\subsection{Numerical Magnitude Relations among Measures}
	Ordering relations explore the consistency of the ordering induced on the set of states by different measures. Even if the orderings may differ, there often exist universal numerical inequalities between the values of different measures. These inequalities quantitatively connect different notions of coherence and provide upper and lower bounds for the measure values.
	\begin{proposition}
		The robustness of block coherence measure $C_{rob}$ has a concise lower bound expression:
		\begin{eqnarray*}
			C_{rob}(\rho,\mathbf{P})\geqslant -\text{Tr}(\rho W_\rho),
		\end{eqnarray*}
		where $W_\rho=\Delta(\rho)-\rho$.
	\end{proposition}
	\begin{proof}
		\par\noindent\par
		For any state $\rho$, we know $\rho\geqslant 0$ and $\Delta(\rho)\geqslant 0$, so it has a spectral decomposition $\Delta(\rho)=\sum_j\lambda_j\ket{j}\bra{j}$. Therefore, $I-\Delta(\rho)=\sum_j(1-\lambda_j)\ket{j}\bra{j}$. This implies $I-\Delta(\rho)\geqslant 0$, and consequently $X_1:= I-W_{\rho}= I-\Delta(\rho)+\rho\geqslant 0$. Simultaneously, $\Delta(X_1)=\Delta(I)-\Delta(\Delta(\rho))+\Delta(\rho)=I$ holds. Thus, $X_1$ can serve as a feasible solution for the semidefinite program for $C_{rob}(\rho,\mathbf{P})$, yielding:
		\begin{align*}
			C_{rob}(\rho,\mathbf{P})&\geqslant\text{Tr}(X_1\rho)-1
			\\
			&=\text{Tr}(\rho)-\text{Tr}(W_\rho\rho)-1
			\\
			&=-\text{Tr}(\rho W_\rho),
		\end{align*}
		That is, the conclusion holds; this is a lower bound for the robustness of block coherence.
	\end{proof}
	\begin{proposition}
		Concerning $C_\text{Tr}(\rho,\mathbf{P})$, the following conclusions hold:
		\par\noindent
		(1) $\widetilde{C}_{l_1}(\rho,\mathbf{Q})\geqslant C_\text{Tr}(\rho,\mathbf{Q})$;
		\par\noindent
		(2) $C_r(\rho,\mathbf{Q})\geqslant \frac{1}{2\ln2}[C_\text{Tr}(\rho,\mathbf{Q})]^2$;
		\par\noindent
		(3) $C_{geo}(\rho,\mathbf{P})\geqslant\frac{1}{4}[C_\text{Tr}(\rho,\mathbf{P})]^2$.
		\par\noindent
		(4) $C_\text{rob}(\rho,\mathbf{P})\geqslant C_\text{Tr}(\rho,\mathbf{P})$.
	\end{proposition}
	\begin{proof}
		\par\noindent\par\noindent
		(1) By the triangle inequality, we have
		\begin{align*}
			\norm{\rho_{ST}}_{\text{Tr}}+\norm{\rho_{TS}}_{\text{Tr}}&\geqslant\norm{\rho_{ST}+\rho_{TS}}_{\text{Tr}}
			\\
			&=\norm{\rho-\Delta(\rho)}_{\text{Tr}}
			\\
			&\geqslant\min_{\lambda>0,\sigma\in\mathcal{I}_{\text{B}}}\norm{\rho-\lambda\sigma}_\text{Tr}.
		\end{align*}
		Thus, the conclusion holds.
		\par\noindent
		(2) For any two quantum states $\rho$ and $\sigma$, the quantum Pinsker inequality \cite{carlen2018some} gives the relationship between relative entropy and the trace norm:
		\begin{eqnarray*}
			\text{Tr}(\rho\log\rho)-\text{Tr}(\rho\log\Delta(\rho))=S(\rho||\sigma)\geqslant\frac{1}{2ln2}\norm{\rho-\sigma}^2_\text{Tr},
		\end{eqnarray*}
		Therefore, $\frac{1}{2ln2}\norm{\rho-\Delta(\rho)}^2_\text{tr}\leqslant S(\rho||\Delta(\rho))=\text{Tr}(\rho\log\rho)-\text{Tr}(\rho\log\Delta(\rho))$. Since
		\begin{align*}
			\text{Tr}[\rho\log\Delta(\rho)]
			&=\begin{gathered}
				\text{Tr}\left[
				\begin{pmatrix}
					\rho_{SS} & \rho_{ST} \\
					\rho_{TS} & \rho_{TT}
				\end{pmatrix}
				\begin{pmatrix}
					\log\rho_{SS} &  \\
					& \log\rho_{TT}
				\end{pmatrix}
				\right]
			\end{gathered}
			\\
			&=\begin{gathered}
				\text{Tr}
				\begin{pmatrix}
					\rho_{SS}\log\rho_{SS} & \rho_{ST}\log\rho_{TT} \\
					\rho_{TS}\log\rho_{SS} & \rho_{TT}\log\rho_{TT}
				\end{pmatrix}
			\end{gathered}
			\\
			&=\begin{gathered}
				\text{Tr}
				\begin{pmatrix}
					\rho_{SS}\log\rho_{SS} &  \\
					& \rho_{TT}\log\rho_{TT}
				\end{pmatrix}
			\end{gathered}
			\\
			&=\begin{gathered}
				\text{Tr}\left[
				\begin{pmatrix}
					\rho_{SS} &  \\
					& \rho_{TT}
				\end{pmatrix}
				\begin{pmatrix}
					\log\rho_{SS} &  \\
					& \log\rho_{TT}
				\end{pmatrix}
				\right]
			\end{gathered}
			\\
			&=\text{Tr}[\Delta(\rho)\log\Delta(\rho)].
		\end{align*}
		And $C_r(\rho,\mathbf{Q})=S(\Delta(\rho)))-S(\rho)=\text{Tr}(\rho\log\rho)-\text{Tr}(\Delta(\rho)\log\Delta(\rho))$, consequently $S(\rho||\Delta(\rho))=C_r(\rho,\mathbf{Q})$, i.e., $C_r(\rho,\mathbf{Q})\geqslant\frac{1}{2ln2}\norm{\rho-\Delta(\rho)}^2_\text{Tr}$.
		\par
		Combined with $\norm{\rho-\Delta(\rho)}^2_\text{Tr}\geqslant\min_{\lambda>0,\sigma\in\mathcal{I}_{\text{B}}}\norm{\rho-\lambda\sigma}_\text{Tr}$, the conclusion follows.
		\par\noindent
		(3) The Fuchs-van de Graaf inequality \cite{zhang2016lower} states that for any two quantum states $\rho$ and $\sigma$, we have
		\begin{eqnarray*}
			1-F^2(\rho,\sigma)\geqslant\frac{1}{4}\norm{\rho-\sigma}^2_{\text{Tr}}.
		\end{eqnarray*}
		Let $\sigma_1$ be the optimal solution achieving the minimum for the semidefinite program of $C_{geo}(\rho,\mathbf{P})$. Then it holds that
		\begin{eqnarray*}
			C_{geo}(\rho,\mathbf{P})=1-F^2(\rho,\Delta(\sigma_1))\geqslant\frac{1}{4}\norm{\rho-\Delta(\sigma_1)}^2_{\text{Tr}}\geqslant\frac{1}{4}[C_\text{Tr}(\rho,\mathbf{P})]^2.
		\end{eqnarray*}
		\par\noindent
		(4) Let $\delta_1$ be the optimal solution achieving the minimum $s_1$ for $C_{rob}(\rho,\mathbf{P})$, i.e., $\rho\leqslant(1+s_1)\Delta(\delta_1)$. Denote $\sigma_1=\Delta(\delta_1)\in\mathcal{I}_{\text{B}}$. Take $\lambda_1=1+s_1$, then $\lambda_1\sigma_1-\rho$ is a positive semidefinite matrix. Its singular values are its eigenvalues, and its trace norm equals its trace. Thus, we have
		\begin{eqnarray*}
			\norm{\rho-\lambda_1\sigma_1}_{\text{Tr}}=\norm{\lambda_1\sigma_1-\rho}_{\text{Tr}}=\text{Tr}(\lambda_1\sigma_1-\rho)=s_1.
		\end{eqnarray*}
		Therefore, $C_{rob}(\rho,\mathbf{P})\geqslant C_\text{Tr}(\rho,\mathbf{P})$.
	\end{proof}
	\begin{proposition}
		$C_r(\rho,\mathbf{P})\leqslant C_{max}(\rho,\mathbf{P})\leqslant\log[1+C_{rob}(\rho,\mathbf{P})]$.
	\end{proposition}
	\begin{proof}
		\par\noindent\par
		Assume for the quantum state $\rho$, $C_{rob}(\rho,\mathbf{P})=s_1$, then there exists $\delta_1\in\mathcal{I}_{\text{B}}$ such that
		\begin{eqnarray*}
			\rho\leqslant(1+s_1)\Delta(\delta_1).
		\end{eqnarray*}
		Let $\sigma_1=\Delta(\delta_1)\in\mathcal{I}_{\text{B}}$, then $\rho\leqslant(1+s_1)\sigma_1$. This leads to
		\begin{eqnarray*}
			\min_{\sigma\in\mathcal{I}_{\text{B}}}\log\min_{\lambda}\left\lbrace \lambda\geqslant 1|\rho\leqslant\lambda\sigma\right\rbrace \leqslant\log\min_{\lambda}\left\lbrace \lambda\geqslant 1|\rho\leqslant\lambda\sigma_1\right\rbrace\leqslant\log(1+s_1).
		\end{eqnarray*}
		That is, $C_{max}(\rho,\mathbf{P})\leqslant\log[1+C_{rob}(\rho,\mathbf{P})]$.
		\par
		The left inequality can be derived from the relationship between the max-relative entropy and the quantum relative entropy \cite{datta2009min}. For any quantum states $\rho$ and $\sigma$, we have $S(\rho||\sigma)\leqslant D_{max}(\rho||\sigma)$. Let $\sigma_2$ be the optimal solution for $C_{max}(\rho,\mathbf{P})$. Then $C_{max}(\rho,\mathbf{P})=D_{max}(\rho||\sigma_2)\geqslant S(\rho||\sigma_2)\geqslant C_r(\rho,\mathbf{P})$, which completes the proof.
	\end{proof}
	\begin{proposition}
		When $\rho$ is a pure state, $\widetilde{C}_{l_1}(\rho,\mathbf{Q})\geqslant C_r(\rho,\mathbf{Q})$ holds.
	\end{proposition}
	\begin{proof}
		\par\noindent\par
		Take any pure state $\ket{\psi}=\alpha_S\ket{S}+\sum_{j=-1}^{1}\beta_j\ket{T_j}$. When $\alpha_S=0$ or $\alpha_S=1$, $\widetilde{C}_{l_1}(\rho,\mathbf{Q})=C_r(\rho,\mathbf{Q})=0$. When $\alpha_S\in(0,1)$, we need to compare the relationship between $f_{l_1}(x)$ and $f_r(x)$ for $x\in(0,1)$. Let $f(x)=f_{l_1}(x)-f_r(x)$ and differentiate to obtain:
		\begin{eqnarray*}
			\frac{\text{d}f}{\text{d}x}=\frac{1-2x}{\sqrt{x(1-x)}}+\log(\frac{x}{1-x}).
		\end{eqnarray*}
		Let $\theta=\sqrt{\frac{x}{1-x}}$, then $\theta\in(0,+\infty)$, $x=\frac{\theta^2}{1+\theta^2}$. The above expression can be written as:
		\begin{eqnarray*}
			\frac{\text{d}f}{\text{d}x}=\frac{1-2\theta^2}{\theta}+2\log(\theta).
		\end{eqnarray*}
		Let $m(\theta)=\ln(2)f'(x)$; then $m(\theta)$ has the same sign as $f'(x)$. Differentiating $m(\theta)$ further:
		\begin{eqnarray*}
			\frac{\text{d}m}{\text{d}x}=\frac{2}{\theta}-\ln(2)-\frac{ln(2)}{\theta^2}.
		\end{eqnarray*}
		Let $n(\theta)=\theta^2m'(\theta)=-\ln(2)\theta^2+2\theta-\ln(2)$; then $n(\theta)$ has the same sign as $m'(\theta)$. Note that $n(\theta)$ is a quadratic function. Its discriminant is $4-4[\ln(2)]^2>0$. We can obtain the two roots of $n(\theta)$: $\theta_1\approx0.403$, $\theta_2\approx2.481$. The following analysis can be made:
		\par\noindent
		(1) When $\theta\in(0,\theta_1)$, $m'(\theta)<0$. This means $m(\theta)$ decreases from $+\infty$ to $m(\theta_1)<0$. There exists a zero $\theta_a\in(0,\theta_1)$ such that $m(\theta_a)=0$;
		\par\noindent
		(2) When $\theta\in(\theta_1,\theta_2)$, $m'(\theta)>0$. This means $m(\theta)$ increases from $m(\theta_1)<0$ to $m(\theta_2)>0$. Since $m(1)=0$, we have $m(\theta)<0$ for $\theta\in(\theta_1,1)$ and $m(\theta)>0$ for $\theta\in(1,\theta_2)$;
		\par\noindent
		(3) When $\theta\in(\theta_2,\infty)$, $m'(\theta)<0$. This means $m(\theta)$ decreases from $m(\theta_2)>0$ to $-\infty$. There exists a zero $\theta_b\in(\theta_2,\infty)$ such that $m(\theta_b)=0$.
		\par
		It follows that $f'(x)$ has the same properties described above. That is, $f'(x)>0$ for $x\in\left(0,\frac{\theta_a^2}{1+\theta_a^2}\right) \cup \left(\frac{1}{2},\frac{\theta_b^2}{1+\theta_b^2}\right)$, and $f'(x)<0$ for $x\in\left(\frac{\theta_a^2}{1+\theta_a^2},\frac{1}{2}\right) \cup \left(\frac{\theta_b^2}{1+\theta_b^2},1\right)$. Since $f(x)\to 0^{+}$ as $x\to 0^{+}$ and $x\to 1^{-}$, and $f(\frac{1}{2})=0$, we can conclude that for $x\in(0,1)$, $f(x)\geqslant0$. That is, $f_{l_1}(x)\geqslant f_r(x)$, and thus $\widetilde{C}_{l_1}(\rho,\mathbf{Q})\geqslant C_r$ holds.
	\end{proof}
	\begin{proposition}
		For S-T coherence measures, when $\alpha>0.5$ and $\rho=\ket{\psi}\bra{\psi}$ is a pure state, the following inequality chain holds:
		\begin{eqnarray*}
			\widetilde{C}_{l_1}(\rho,\mathbf{Q})\geqslant 2\text{Var}(\rho,Q_S)\geqslant C_{\alpha,1}(\rho,\mathbf{Q}),
		\end{eqnarray*}
		where $\text{Var}(\rho,A)=\text{Tr}(\rho A^2)-[\text{Tr}(\rho A)]^2$ is the variance of operator $A$ with respect to $\rho$.
	\end{proposition}
	\begin{proof}
		\par\noindent\par
		We prove the second inequality. Since $Q_S^2=Q_S$, the variance expressed with the projector $Q_S$ is $\text{Var}(\rho,Q_S)=\text{Tr}(\rho Q_S)-[\text{Tr}(\rho Q_S)]^2$. Note that $\text{Var}(\rho,Q_S)$ and $I_{WY}(\rho,Q_S)=\text{Tr}(\rho Q_S)-\text{Tr}(\sqrt{\rho}Q_S\sqrt{\rho}Q_S)$ share the same term $\text{Tr}(\rho Q_S)$. Let us first prove $\text{Var}(\rho,Q_S)\geqslant I_{WY}(\rho,Q_S)$, i.e., prove $\text{Tr}(\sqrt{\rho}Q_S\sqrt{\rho}Q_S)\geqslant[\text{Tr}(\rho Q_S)]^2$. In fact, this inequality holds for any state $\rho$. Performing a spectral decomposition of $\rho$, $\rho=\sum_i\lambda_i\ket{i}\bra{i}$, we have
		\begin{align*}
			\text{Tr}(\sqrt{\rho}Q_S\sqrt{\rho}Q_S)
			&=\text{Tr}\left(\sum_i\lambda_i\ket{i}\bra{i}Q_S\sum_j\lambda_j\ket{j}\bra{j}Q_S \right)
			\\
			&=\sum_{i,j}\sqrt{\lambda_i\lambda_j}|\bra{i}Q_S\ket{j}|^2
			\\
			&=\sum_{i}\lambda_i|\bra{i}Q_S\ket{j}|^2+\sum_{i\neq j}\sqrt{\lambda_i\lambda_j}|\bra{i}Q_S\ket{j}|^2,
		\end{align*}
		and
		\begin{eqnarray*}
			[\text{Tr}(\rho Q_S)]^2=\sum_i\lambda_i|\bra{i}Q_S\ket{i}|.
		\end{eqnarray*}
		Let $a_i=\sqrt{\lambda_i}|\bra{i}Q_S\ket{i}|$, $b_i=\sqrt{\lambda_i}$. By the H{\"o}lder inequality, we have:
		\begin{align*}
			\sum_i\lambda_i|\bra{i}Q_S\ket{i}|&=\sum_ia_ib_i
			\\
			&\leqslant\left(\sum_ia_i^2 \right)^{1/2} \left(\sum_ib_i^2 \right)^{1/2}
			\\
			&=\left( \sum_i\lambda_i|\bra{i}Q_S\ket{i}|^2\right)^{1/2} \left(\sum_i\lambda_i \right)^{1/2}
			\\
			&=\left( \sum_i\lambda_i|\bra{i}Q_S\ket{i}|^2\right)^{1/2},
		\end{align*}
		Thus, we have $[\text{Tr}(\rho Q_S)]^2\leqslant\sum_{i}\lambda_i|\bra{i}Q_S\ket{j}|^2\leqslant\text{Tr}(\sqrt{\rho}Q_S\sqrt{\rho}Q_S)$. This indicates $2\text{Var}(\rho,Q_S)\geqslant C_{WY}(\rho,\mathbf{Q})$. From the previous proposition, we know $2\text{Var}(\rho,Q_S)\geqslant C_{\alpha,1}(\rho,\mathbf{Q})$.
		\par
		We prove the first inequality. Since $\rho=\ket{\psi}\bra{\psi}$ is a pure state, we have
		\begin{align*}
			\text{Var}(\rho,Q_S)&=\text{Tr}(\rho Q_S)-[\text{Tr}(\rho Q_S)]^2
			\\
			&=\bra{psi}Q_S\ket{\psi}-(\bra{psi}Q_S\ket{\psi})^2
			\\
			&=\alpha_S^2(1-\alpha_S^2)
			\\
			&=x(1-x).
		\end{align*}
		For a pure state, $\widetilde{C}_{l_1}(\rho,\mathbf{Q})=f_{l_1}(x)=2\sqrt{x(1-x)}$, hence $\widetilde{C}_{l_1}(\rho,\mathbf{Q})\geqslant 2\text{Var}(\rho,Q_S)$.
	\end{proof}
	\begin{proposition}
		When $\alpha=\beta$, the following holds: $\frac{1}{1-\alpha}C_{\alpha,z}(\rho,\mathbf{P})\leqslant C^T_{\alpha}(\rho,\mathbf{P})\leqslant C^N_{\alpha}(\rho,\mathbf{P})$.
	\end{proposition}
	\begin{proof}
		\par\noindent\par
		In the proof of Theorem 3.4, we obtained the following inequality:
		\begin{align*}
			\text{Tr}(\rho\sharp_{1-\alpha}\sigma)\leqslant\text{Tr}(\rho^{\alpha}\sigma^{1-\alpha})\leqslant\text{Tr}(\sigma^{\frac{1-\alpha}{2z}}\rho^{\frac{\alpha}{z}}\sigma^{\frac{1-\alpha}{2z}})^z.
		\end{align*}
		Then for $C_{\alpha,z}(\rho,\mathbf{P})$, there exists $\gamma\in\mathcal{I}_{\text{B}}$ such that $C_{\alpha,z}(\rho,\mathbf{P})=1-\big[\text{Tr}(\gamma^{\frac{1-\alpha}{2z}}\rho^{\frac{\alpha}{z}}\gamma^{\frac{1-\alpha}{2z}})^z \big]^{\frac{1}{\alpha}}$. Thus, we have:
		\begin{align*}
			C_{\alpha,z}(\rho,\mathbf{P})
			&=1-\big[\text{Tr}(\gamma^{\frac{1-\alpha}{2z}}\rho^{\frac{\alpha}{z}}\gamma^{\frac{1-\alpha}{2z}})^z\big]^{\frac{1}{\alpha}}
			\\
			&\leqslant 1-\big[\text{Tr}(\rho^{\alpha}\sigma^{1-\alpha})\big]^{\frac{1}{\alpha}}
			\\
			&\leqslant 1-\max_{\sigma\in\mathcal{I}_{\text{B}}}\big[\text{Tr}(\rho^{\alpha}\sigma^{1-\alpha})\big]^{\frac{1}{\alpha}}
			\\
			&=(1-\alpha)C^T_{\alpha}(\rho,\mathbf{P}).
		\end{align*}
		For $C^T_{\alpha}(\rho,\mathbf{P})$, there similarly exists $\xi\in\mathcal{I}_{\text{B}}$ such that $C^T_{\alpha}(\rho,\mathbf{P})=1-\big[\text{Tr}(\rho^{\alpha}\xi^{1-\alpha})\big]^{\frac{1}{\alpha}}$. Hence:
		\begin{align*}
			C^T_{\alpha}(\rho,\mathbf{P})
			&=1-\big[\text{Tr}(\rho^{\alpha}\xi^{1-\alpha})\big]^{\frac{1}{\alpha}}
			\\
			&\leqslant 1-\big[\text{Tr}(\rho\sharp_{1-\alpha}\xi)\big]^{\frac{1}{\alpha}}
			\\
			&\leqslant 1-\max_{\sigma\in\mathcal{I}_{\text{B}}}\big[\text{Tr}(\rho\sharp_{1-\alpha}\sigma)\big]^{\frac{1}{\alpha}}
			\\
			&=C^N_{\alpha}(\rho,\mathbf{P}).
		\end{align*}
		Thus, the conclusion is established.
	\end{proof}

	\section{Dynamics of S-T Coherence Measures under the Kominis Master Equation}
	In many physical, chemical, and biological systems (e.g., energy transfer in photosynthesis, radical pair reactions), the open system dynamics of quantum systems are often described by master equations. Many master equations indeed have non-trace-preserving characteristics. This section primarily investigates the time evolution under the Kominis master equation that includes a recombination term.
	\par
	The particle number (population) of the system refers to the trace of the density matrix, i.e., $\text{Tr}(\rho)$, which represents the total number of radical pairs (RPs) in the ensemble (or the total probability). For a normalized initial state, $\text{Tr}(\rho)=1$ indicates one unit of RP (can be understood as an ensemble). Under the evolution of a master equation containing a recombination term, the population decays over time, with the physical meaning being the proportion of radical pairs that have not yet undergone recombination and are still alive.
	\par
	In the previous section, as we aimed to study the decay of S-T coherence under the amplitude damping channel, the chosen initial pure states were all S-T coherent states. When considering real radical pair (RP) reactions, the system not only experiences decoherence but also undergoes spin-dependent recombination reactions, leading to a decrease in the system's particle number. To more realistically simulate the dynamics of RPs, we need to adopt a complete master equation that includes Hamiltonian evolution, decoherence, and recombination terms. Combined with the S-T coherence measure $C_{\alpha,1}$ we defined earlier, we study the evolution of S-T coherence in RPs throughout the complete physical process and analyze how coherence influences the reaction yield. We simulate the complete dynamics of radical pairs (RPs) based on the master equation framework proposed by Kominis \cite{kominis2015radical}. In this model, the time evolution of the RP density matrix $\rho$ is described by the following master equation:
	\begin{eqnarray*}
		\frac{\text{d}\rho}{\text{d}t}=-i[\mathcal{H},\rho]+\mathcal{D}\llbracket\rho\rrbracket+\mathcal{R}_K\llbracket\rho\rrbracket,
	\end{eqnarray*}
	where $\mathcal{H}$ is the Hamiltonian, $\mathcal{D}\llbracket\rho\rrbracket$ describes S-T decoherence, and $\mathcal{R}_K\llbracket\rho\rrbracket$ describes recombination.
	\par
	\par
	For the Hamiltonian $\mathcal{H}$, we choose a simple form that can intuitively demonstrate the changes:
	\begin{eqnarray*}
		\mathcal{H}=\omega_1s_{1z}\otimes I_2+\omega_2 I_1\otimes s_{2z},
	\end{eqnarray*}
	where $s_{1z}$ is the $z$-component Pauli operator of the $i$-th electron spin (using natural units $\hbar=1$). $\omega_1$ and $\omega_2$ are the Larmor frequencies, determining the period and strength of the evolution. It is worth noting that under the action of this Hamiltonian $\mathcal{H}$, the dynamics of the entire system are decoupled into two independent subsystems: one is the "inactive" subspace composed of $\ket{T_{-1}}$ and $\ket{T_1}$, and the other is the "active" subspace consisting of $\ket{S}$ and $\ket{T_0}$, where coherent oscillations occur. In fact, the Hamiltonian $\mathcal{H}$ can be written as:
	\begin{eqnarray*}
		\mathcal{H}=\frac{\omega_1-\omega_2}{2}(\ket{S}\bra{T_0}+\ket{T_0}\bra{S})+\frac{\omega_1+\omega_2}{2}(\ket{T_1}\bra{T_1}-\ket{T_{-1}}\bra{T_{-1}}),
	\end{eqnarray*}
	When any two states $\ket{\psi}$ and $\bra{\phi}$ belong to the two subspaces $\left\lbrace \ket{T_1},\ket{T_{-1}} \right\rbrace $ and $\left\lbrace \ket{S},\ket{T_0} \right\rbrace $, respectively, it can be verified that $\bra{\psi}\mathcal{H}\ket{\phi}=\bf{0}$. Therefore, the entire evolution is confined within $\left\lbrace \ket{S},\ket{T_0} \right\rbrace $, meaning that the coupling between $\ket{S}$ and $\ket{T_0}$ driven by the Hamiltonian $\mathcal{H}$ is the sole source of S-T coherence. Simultaneously, the evolution term $-i[\mathcal{H},\rho]$ also maintains the unitarity of the system, i.e., $\text{Tr}\left\lbrace -i[\mathcal{H},\rho] \right\rbrace =0$, which means the evolution term does not change the particle number of the system.
	\par
	The decoherence term $\mathcal{D}\llbracket\rho\rrbracket$, originating from the interaction between the RP and the environment, has the form:
	\begin{eqnarray*}
		\mathcal{D}\llbracket\rho\rrbracket=-K_d(Q_S\rho Q_T+Q_T\rho Q_S),
	\end{eqnarray*}
	where $K_d=\frac{k_S+k_T}{2}$ is the decoherence rate, and $k_S$ and $k_T$ are the recombination rate constants for the singlet and triplet, respectively. During time evolution, the decoherence term $\mathcal{D}\llbracket\rho\rrbracket$ continuously and irreversibly eliminates the S-T coherence in the density matrix (the off-diagonal elements $\rho_{ST}$ and $\rho_{TS}$). However, from $\text{Tr}(\mathcal{D}\llbracket\rho\rrbracket)=0$, we know that the decoherence term does not change the particle number of the system.
	\par
	The recombination term $\mathcal{R}_K\llbracket\rho\rrbracket$ describes the process by which RPs generate neutral products through charge recombination, a process that reduces the total number of RPs. Kominis pointed out that for a density matrix $\rho$, there exists a decomposition that yields a maximal incoherent state $\rho_{\text{incoh}}$ and a maximal coherent state $\rho_{\text{coh}}$:
	\begin{align*}
		\rho=(1-p_{\text{coh}})\rho_{\text{incoh}}+p_{\text{coh}}\rho_{\text{coh}},
	\end{align*}
	where $p_{\text{coh}}$ is related to the chosen S-T coherence measure, representing the strength of the system's S-T coherence, with a value range of $[0,1]$. $\rho_{\text{incoh}}=\rho_{SS}+\rho_{TT}$, $\rho_{\text{coh}}=\rho_{SS}+\rho_{TT}+\frac{1}{p_{\text{coh}}}(\rho_{ST}+\rho_{TS})$. We adopt the coherence-dependent recombination term proposed by Kominis:
	\begin{align*}
		\mathcal{R}_K\llbracket\rho\rrbracket
		&=-(1-p_{\text{coh}})(k_S\rho_{SS}+k_T\rho_{TT})-p_{\text{coh}}\frac{\text{d}r_S+\text{d}r_T}{\text{d}t}\frac{\rho_{\text{coh}}}{\text{Tr}(\rho)},
	\end{align*}
	where $\text{d}r_S=k_S\text{d}t\text{Tr}(\rho Q_S)$ and $\text{d}r_T=k_T\text{d}t\text{Tr}(\rho Q_T)$ are the instantaneous recombination rates (or instantaneous yields) for the singlet and triplet, respectively, i.e., the rate at which products are generated per unit time via recombination. To avoid singularities when the measure yields $p_{\text{coh}}=0$, we write the recombination term $\mathcal{R}_K\llbracket\rho\rrbracket$ as:
	\begin{align*}
		\mathcal{R}_K\llbracket\rho\rrbracket=
		&-(1-p_{\text{coh}})(k_S\rho_{SS}+k_T\rho_{TT})
		\\
		&-\frac{k_S\text{Tr}(\rho Q_S)+k_T\text{Tr}(\rho Q_T)}{\text{Tr}(\rho)}[p_{\text{coh}}(\rho_{SS}+\rho_{TT})+\rho_{ST}+\rho_{TS}].
	\end{align*}
	When $p_{\text{coh}}=0$, the recombination term reduces to the traditional Haberkorn form, also known as the Jones-Hore master equation \cite{kominis2015radical,jones2010spin}, which removes the incoherent part from the density matrix via projection. When $p_{\text{coh}}=1$, the recombination term indicates that the entire coherent state is removed with a certain probability.
	\par
	Since
	\begin{align*}
		\text{Tr}(\mathcal{R}_K\llbracket\rho\rrbracket)=
		&-\text{Tr}[(1-p_{\text{coh}})(k_S\rho_{SS}+k_T\rho_{TT})]
		\\
		&-\text{Tr}\left\lbrace \frac{k_S\text{Tr}(\rho Q_S)+k_T\text{Tr}(\rho Q_T)}{\text{Tr}(\rho)}[p_{\text{coh}}(\rho_{SS}+\rho_{TT})+\rho_{ST}+\rho_{TS}]\right\rbrace
		\\
		=&-(1-p_{\text{coh}})[k_S\text{Tr}(\rho Q_S)+k_T\text{Tr}(\rho Q_T)]
		\\
		&-p_{\text{coh}}[k_S\text{Tr}(\rho Q_S)+k_T\text{Tr}(\rho Q_T)]
		\\
		=&-[k_S\text{Tr}(\rho Q_S)+k_T\text{Tr}(\rho Q_T)].
	\end{align*}
	Thus, we have:
	\begin{align*}
		\frac{\text{d}}{\text{d}t}\text{Tr}(\rho)=\text{Tr}(\mathcal{R}_K\llbracket\rho\rrbracket)=-\frac{\text{d}r_S+\text{d}r_T}{\text{d}t}.
	\end{align*}
	Therefore, the master equation is not a trace-preserving process; it correctly describes the decrease in the particle number of the radical system, and the rate of decrease is exactly equal to the sum of the generation rates of the singlet and triplet reaction products. Then, during the evolution of the density matrix $\rho$, for the calculation of the measure values, we need to use $\hat{\rho}=\frac{\rho}{\text{Tr}(\rho)}$. This allows us to connect the defined S-T coherence measure $C_{\alpha,1}$ with the recombination term. Define the effective coherence parameter as:
	\begin{eqnarray*}
		p_{\text{coh}}^{\text{eff}}=\frac{C_{\alpha,1}(\hat{\rho},\mathbf{Q})}{\max(C_{\alpha,1})},
	\end{eqnarray*}
	where $\max(C_{\alpha,1})$ represents the maximum value of this measure for a given parameter $\alpha$ (e.g., when $\alpha=0.5$, $\max(C_{\alpha,1})=0.5$). In this way, we embed a rigorous resource-theoretic measure into the evolution model, making the effect of the recombination process directly dependent on the coherence strength quantified by our measure. In the following simulations, we set $\alpha=0.5,\omega_1 = 0.8,\omega_2 = 0.3,k_S=0.2,k_T=0.05,t=40$.
	\par
	First, for three evolution scenarios:
	\par\noindent
	Scenario A: The master equation contains only Hamiltonian evolution, corresponding to an ideal closed system;
	\par\noindent
	Scenario B: The master equation contains Hamiltonian evolution and decoherence evolution, corresponding to an open system with decoherence but no particle number loss;
	\par\noindent
	Scenario C: The complete master equation evolution, corresponding to the real RP system in Kominis' theory, which includes both decoherence and recombination evolution.
	\par
	When the initial state is chosen as $\ket{\psi_1}=\ket{S}$, the changes in the measure $C_{0.5,1}$ under the three scenarios can be obtained. Since scenarios A and B have no particle number decay, we only consider the particle number change in scenario C, as shown in the figure below:
	\begin{figure}[H]
		\centering
		\includegraphics[width=0.91\linewidth,height=0.2\textheight]{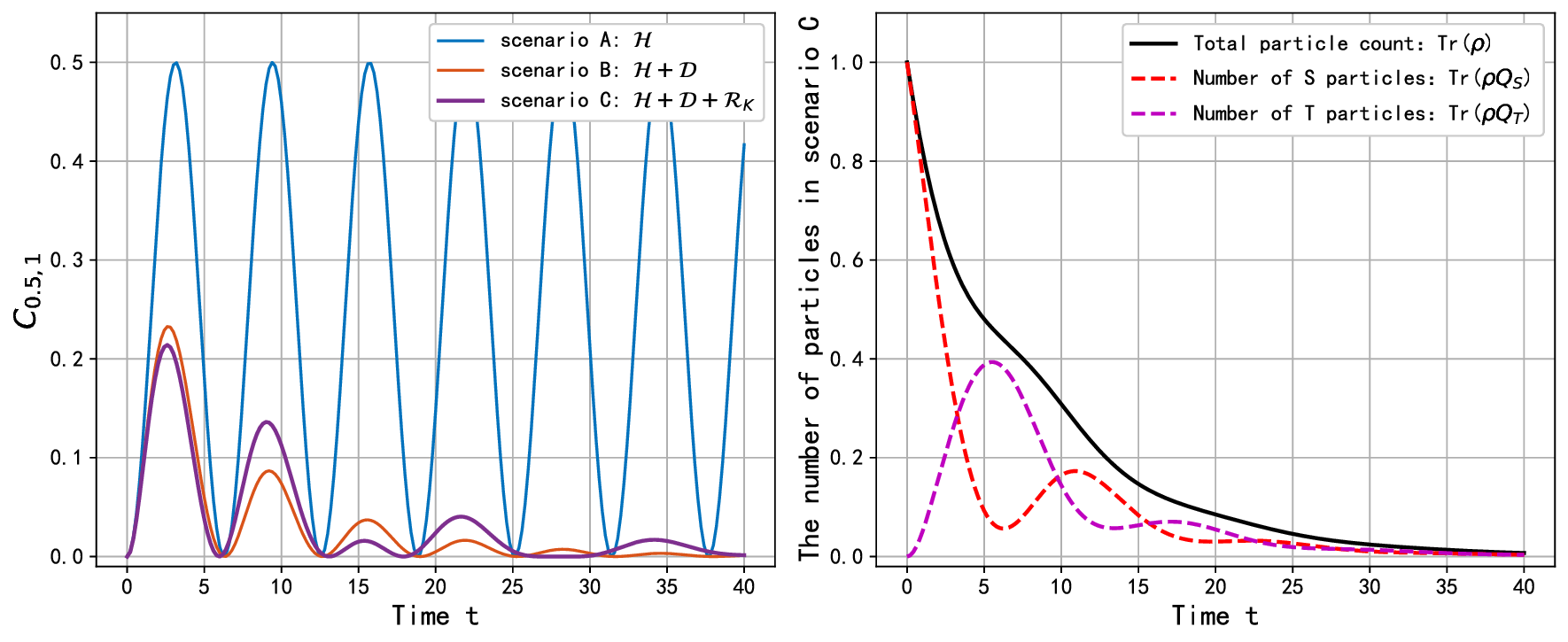}
		\caption{Coherence measure $C_{0.5,1}$ and particle number (initial state: $\ket{S}$)}
	\end{figure}
	\par
	In scenario A, the coherence exhibits perfect periodic oscillations. Let the Larmor frequency difference be $\Omega=\omega_1-\omega_2$. The calculated period $\frac{\pi}{|\Omega|}\approx 0.6283$ matches the figure, with the maximum reaching $\max(C_{0.5,1})=0.5$ and the minimum being zero. This intuitively reflects the role of the Hamiltonian as the "engine" of coherence. The Hamiltonian is the sole source of S-T coherence, continuously driving coherent oscillations between $\ket{S}$ and $\ket{T_0}$, causing the system to transition from the initial singlet state (coherence 0) to maximum coherence (coherence 0.5), then to the triplet state (coherence 0), and repeating this process periodically, with coherence regularly generated and destroyed.
	\par
	In scenario B, coherence still exhibits oscillations, but whenever coherence is generated, the decoherence term counteracts part of it, leading to damped oscillations, so the amplitude decays over time $t$, eventually approaching zero. It can be seen that the trajectory of scenario B is always enclosed within the oscillation envelope of scenario A, confirming the purely destructive effect of the decoherence term on coherence.
	\par
	In scenario C, the dynamics become extremely rich due to the introduction of the recombination term $\mathcal{R}_K\llbracket\rho\rrbracket$. Its evolution trajectory significantly deviates from the simple damped oscillation pattern and exhibits complex crossings with the curves of scenarios A and B. This is because the recombination term introduces a nonlinear feedback mechanism: the recombination process not only removes particles but also its specific manner depends on the instantaneous coherence (via the effective coherence parameter $p_{\text{coh}}^{\text{eff}}$). During periods of high coherence, recombination tends to "remove the entire" quantum state, which not only reduces the particle number but also influences the coherence dynamics of the surviving ensemble through coherence-dependent feedback, thereby altering the composition of the surviving state. During periods of low coherence, it removes more incoherent parts in a classical projection manner, attempting to mitigate the oscillatory decay of coherence. The figure shows multiple instances where the measure value in scenario C is higher than in scenario B, and even higher than in scenario A. This is precisely the result of coherence-dependent feedback, dynamically competing with the driving effect of the Hamiltonian and the destructive effect of decoherence, leading to complex non-monotonic behavior in the coherence evolution trajectory.
	\par
	The decay curve of the particle number $\text{Tr}(\rho)$ further reveals the details of the dynamic competition. At $t=0$, the system consists entirely of singlet particles, and there is no coherence in the system. At this moment, the recombination term immediately removes singlet particles at the rate $k_S$ to generate products. Simultaneously, the Hamiltonian begins to convert the singlet $\ket{S}$ into the triplet $\ket{T_0}$, generating S-T coherence in this process, which in turn triggers damping by the decoherence term. Therefore, in the early stages of evolution, the Hamiltonian, decoherence term, and recombination term work together, causing a sharp decrease in the singlet particle number. However, an interesting phenomenon follows: although the instantaneous singlet recombination yield $k_S=0.2$ is four times that of the triplet $k_T=0.05$, the particle number in scenario C does not show that the singlet decay rate is always higher than the triplet's, nor does the gap gradually widen; instead, they tend to converge during evolution. This is precisely the mediating effect of the Hamiltonian-driven S-T oscillation, as can be seen from the following figure:
	\begin{figure}[H]
		\centering
		\includegraphics[width=0.91\linewidth,height=0.2\textheight]{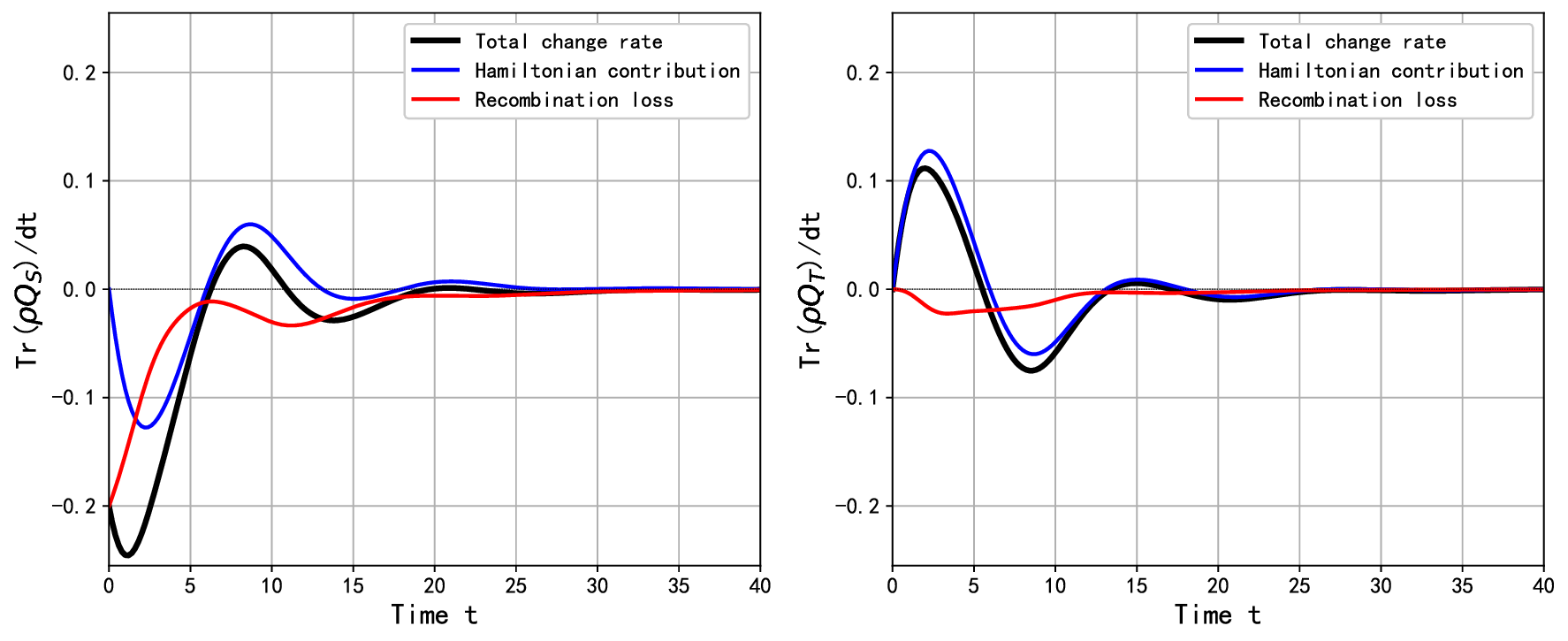}
		\caption{Singlet and Triplet Particle Number Change Rates (initial state: $\ket{S}$)}
	\end{figure}
	\par
	As the overall singlet particle number continuously decreases and the triplet particle number increases, a dynamic feedback forms between this asymmetric consumption and the continuous conversion by the Hamiltonian: when the singlet particle number is too low, the Hamiltonian tends to convert triplets back to singlets; and vice versa. Consequently, the system exhibits a decrease in the singlet particle decay rate and an increase in the triplet particle decay rate, and this process repeats, causing particles to continuously convert between singlet and triplet states, compensating for the particle number difference caused by the asymmetric recombination rates. The result of this dynamic balance is that, despite the intrinsic recombination rates differing by a factor of four, the singlet and triplet particle numbers do not rapidly diverge but remain of the same order of magnitude during evolution until the system is nearly depleted. At the end of the time evolution, since a large number of particles have been converted into neutral products by the recombination term, the total particle number $\text{Tr}(\rho)$ of the system approaches 0, and the cumulative yield $Y_S+Y_T$ approaches $100\%$.
	\par
	As the evolution time $t$ increases, in the general trend, the singlet particle number $\text{Tr}(\rho Q_S)$ and the triplet particle number $\text{Tr}(\rho Q_T)$ continuously decrease, and the corresponding instantaneous yields $k_S\text{Tr}(\rho Q_S)$ and $k_T\text{Tr}(\rho Q_T)$ also continuously decrease. This completely conforms to the basic laws of chemical kinetics, with the root cause being the decay of reactant concentration. This is why the yields from all sides tend to stabilize at the end of the time evolution in the figure.
	\begin{figure}[H]
		\centering
		\includegraphics[width=0.91\linewidth,height=0.2\textheight]{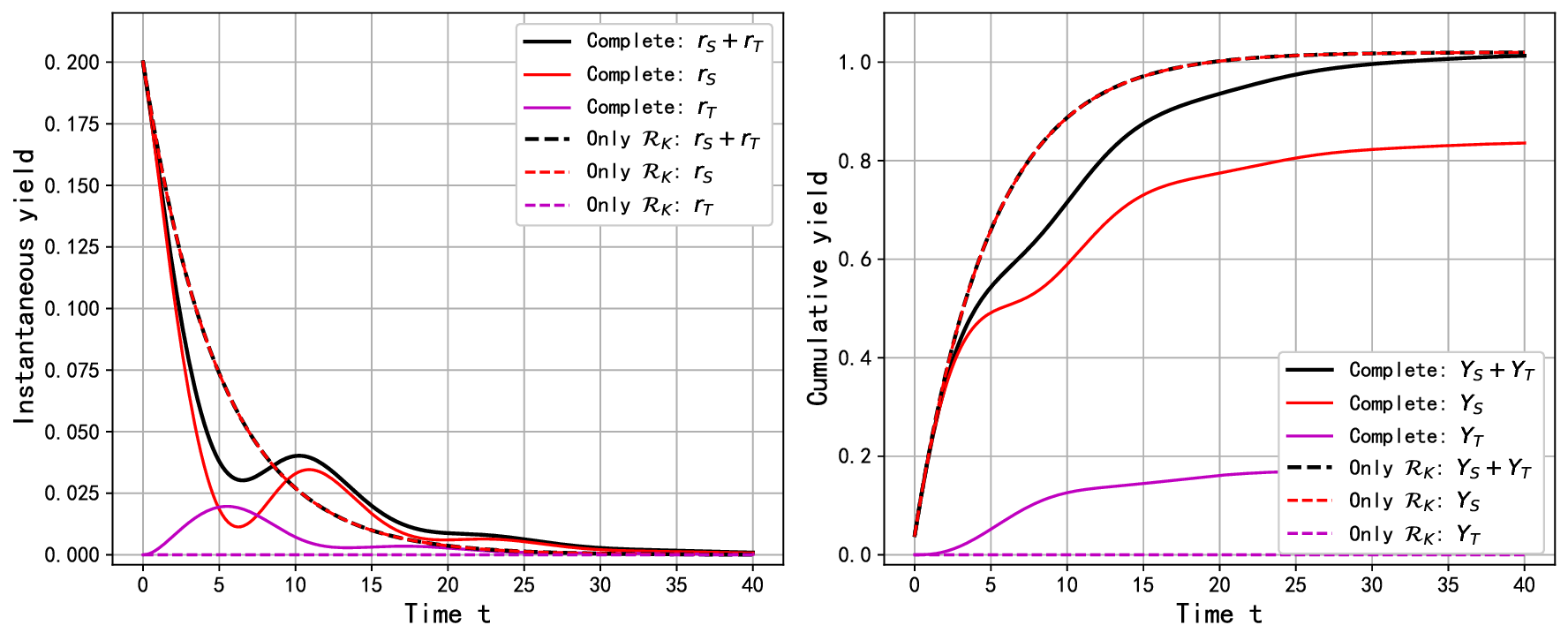}
		\caption{Coherence measure $C_{0.5,1}$ and particle number (initial state: $\ket{S}$)}
	\end{figure}
	\par
	For comparison, we also simulated the scenario where the master equation contains only the recombination term, i.e., $\mathcal{H}=\bf{0}$, $\mathcal{D}=\bf{0}$. In this scenario, since there is no Hamiltonian to drive S-T conversion, the initial singlet remains a singlet throughout, so the triplet instantaneous yield $r_S$ and cumulative yield $Y_T$ are always zero. This contrasts sharply with the significant non-zero triplet yield produced by the rich dynamics under the complete master equation, thus numerically confirming the indispensable key role of the Hamiltonian (and the coherence it drives) in initiating the triplet reaction channel and determining the final product ratio.
	\par
	Regarding the Larmor frequency in the Hamiltonian, its role permeates the entire process of coherence generation, evolution, and its final conversion into observable chemical signals. The Larmor frequency difference $\Omega$ not only sets the time scale (frequency and period) of coherent oscillations but also determines the coupling strength of the Hamiltonian driving force. The larger $\Omega$ is, the stronger the Hamiltonian's ability to mix singlet and triplet states, and the faster the system deviates from the initial state and reaches maximum coherence. This directly affects the initial slope of the coherence curve. Simultaneously, in the competition among the three terms of the master equation (Hamiltonian evolution, decoherence term, recombination term), through multiple parameter adjustments, we have the following findings:
	\par\noindent
	(1) If $\Omega \gg k_S,k_T$, then Hamiltonian evolution dominates the process. The system has sufficient time to complete multiple full coherent oscillations before experiencing significant decoherence or recombination. In the experimental curves, multiple oscillations can be observed.
	\par\noindent
	(2) If $\Omega \ll k_S,k_T$, then the decoherence term or recombination term dominates the process. Coherence is quenched before it can be generated, or rapidly decreases after generation and does not rise again. This prevents the system's quantum coherence from being effectively established, and the oscillation effect is strongly suppressed or even completely invisible.
	\par\noindent
	This result is indeed consistent with expectations and also demonstrates again the important role of Hamiltonian driving.
	\par
	Next, we select the initial state $\ket{\psi_2}=\frac{\ket{S}+\ket{T_0}}{\sqrt{2}}$ for simulation analysis. This state is in a maximum coherent superposition of $\ket{S}$ and $\ket{T_0}$, with equal initial singlet and triplet particle numbers. The simulation parameters and code are exactly the same as before; only key results are presented here.
	\begin{figure}[H]
		\centering
		\includegraphics[width=0.91\linewidth,height=0.2\textheight]{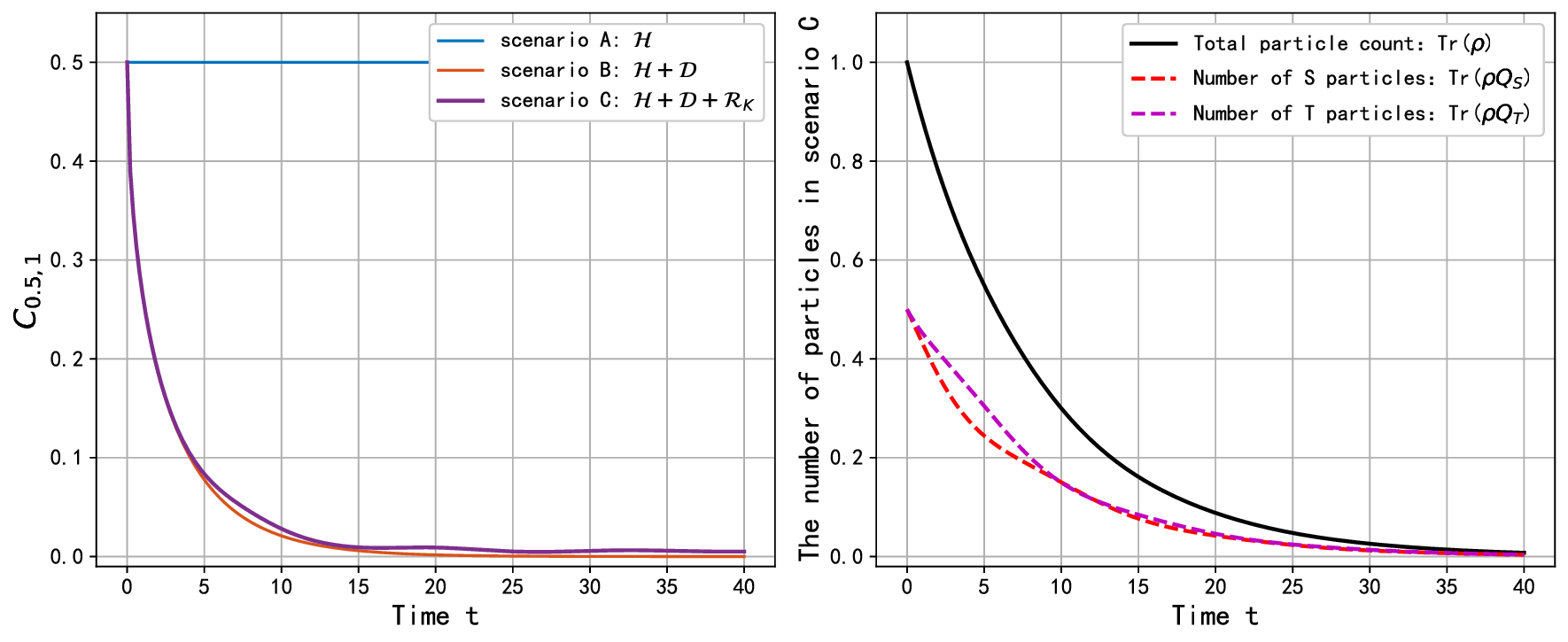}
		\caption{Coherence measure $C_{0.5,1}$ and particle number (initial state: $\frac{\ket{S}+\ket{T_0}}{\sqrt{2}}$)}
	\end{figure}
	\par
	As shown in the figure above, in the ideal closed system with only Hamiltonian evolution (scenario A), the coherence measure $C_{0.5,1}$ no longer exhibits periodic oscillations but remains a horizontal straight line at the maximum value $\max(C_{0.5,1})=0.5$. In fact, due to the equal probability distribution of the initial state in the $\ket{S}$ and $\ket{T_0}$ subspaces, we have:
	\begin{eqnarray*}
		\mathcal{H}\frac{\ket{S}+\ket{T_0}}{\sqrt{2}}=\frac{\omega_1-\omega_2}{2}\left( \frac{\ket{S}+\ket{T_0}}{\sqrt{2}}\right).
	\end{eqnarray*}
	Therefore, $\ket{\psi_2}$ is an eigenstate of the Hamiltonian operator $\mathcal{H}$. The time evolution operator acting on it only produces a global phase factor: $\ket{\psi_2(t)}=e^{-i\frac{\omega_1-\omega_2}{2}t}\ket{\psi_2(0)}$. Thus, we have:
	\begin{eqnarray*}
		\ket{\psi_2(t)}\bra{\psi_2(t)}=\left( e^{-i\frac{\omega_1-\omega_2}{2}t}\ket{\psi_2(0)}\right)\left( e^{i\frac{\omega_1-\omega_2}{2}t}\bra{\psi_2(0)}\right)=\ket{\psi_2(0)}\bra{\psi_2(0)} ,
	\end{eqnarray*}
	That is, with only Hamiltonian evolution (scenario A), the coherence measure $C_{0.5,1}$ remains a constant maximum line. This contrasts sharply with the perfect oscillations produced by the initial state $\ket{S}$ (not an eigenstate of $\mathcal{H}$) in scenario A, visually demonstrating how the eigenstate property of the initial state relative to the Hamiltonian fundamentally determines its coherence evolution pattern. This is also why, although oscillations still occur under the combined action of the Hamiltonian, decoherence term, and recombination term, the fluctuations are noticeably more moderate compared to those for $\ket{S}$ in each figure. The change rates of singlet and triplet particle numbers also well prove this moderateness.
	\par
	We can also see that for $\ket{\psi_2}$,
	\begin{figure}[H]
		\centering
		\includegraphics[width=0.91\linewidth,height=0.2\textheight]{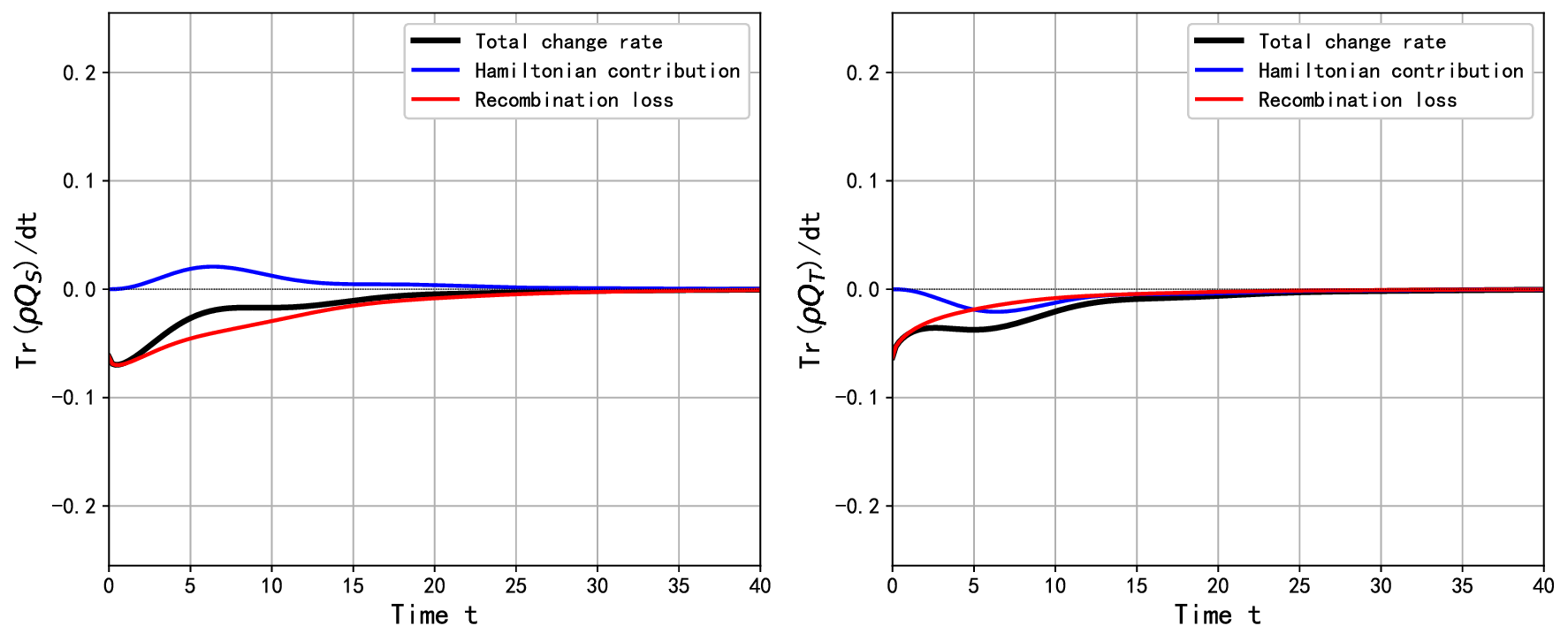}
		\caption{Singlet and Triplet Particle Number Change Rates (initial state: $\frac{\ket{S}+\ket{T_0}}{\sqrt{2}}$)}
	\end{figure}
	\begin{figure}[H]
		\centering
		\includegraphics[width=0.91\linewidth,height=0.2\textheight]{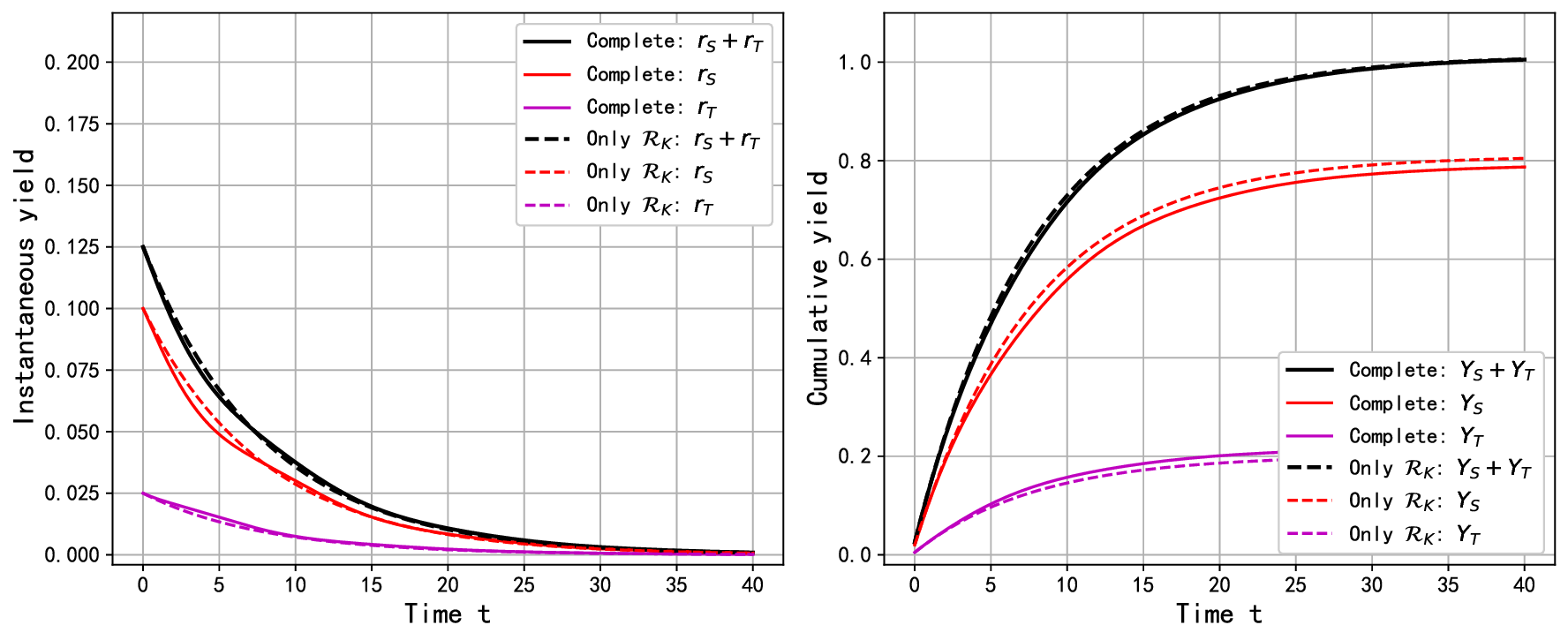}
		\caption{Coherence measure $C_{0.5,1}$ and particle number (initial state: $\frac{\ket{S}+\ket{T_0}}{\sqrt{2}}$)}
	\end{figure}
	\par
	For this initial state, the simulated final singlet product yield $Y_S$ and triplet product yield $Y_T$ have a ratio of approximately 4, very close to the ratio of the instantaneous yield coefficients. For the initial state $\ket{S}$, the ratio is 4.6, also relatively close. This appearance might tempt one to conjecture that the ratio of product yields being close to the ratio of instantaneous yield coefficients is a universal rule.
	\par
	However, we must point out that this conjecture is incorrect. The core innovation of the Kominis model lies precisely in its coherence-dependent recombination term. Under this model, the product yields are the integrals of $k_S\text{Tr}(\rho Q_S)$ and $k_T\text{Tr}(\rho Q_T)$ over the entire nonlinear evolution history, and the evolution of $\text{Tr}(\rho Q_S)$ and $\text{Tr}(\rho Q_T)$ themselves strongly depends on the effective coherence parameter $p_{\text{coh}}^{\text{eff}}$, i.e., the change in the measure value. This itself is a highly nonlinear coupled system. For example, taking the initial state as $\frac{\ket{S}+\ket{T_1}}{\sqrt{2}}$, the simulated singlet product yield $Y_S$ is even significantly lower than the triplet product yield $Y_T$ under the same parameters as in this experiment. For rigor, we selected 10,000 random pure states and simulated them under the same parameters as in this experiment. The results still showed no correlation with the ratio of instantaneous yield coefficients. This indicates that the final yield ratio depends on various factors such as the choice of Hamiltonian, initial state, and coefficients, and is not solely determined by the instantaneous yield coefficients.

	\section{Conclusion}
	This paper conducts systematic research centered on block coherence resource theory, and the main achievements are as follows:
	\par
	In Chapter 3, after presenting the framework of block coherence resource theory, we proved that the set of block coherence measures is closed under convex combination (Theorem 3.1) and utilized the convex roof extension method to construct new measures from any given block coherence measure (Theorem 3.2). This provides a general algebraic scheme for generating new measures from known ones.
	\par
	In Chapter 3, based on the more general $\alpha$-$z$ Rényi relative entropy, we defined a two-parameter family of block coherence measures $C_{\alpha,z}$ (Theorem 3.3). Subsequently, we investigated the properties of $C_{\alpha,z}$, including its dependence on parameters (Proposition 3.1) and its behavior under block-diagonal unitary transformations, direct product, and direct sum (Proposition 3.2). Finally, we specifically applied the $C_{\alpha,z}$ measure family to S-T coherence, defining the corresponding S-T coherence measure (Proposition 3.3). Furthermore, we defined a block coherence measure family $C^N_{\beta}$ via the Tsallis relative operator entropy (Theorem 3.4).
	\par
	In Chapter 4, we reviewed various block coherence measures that have been proposed in block coherence and S-T coherence theory. Through analytical derivations and numerical counterexamples, we focused on studying the consistency of ordering induced on quantum state sets by different measures. It was proven that for pure states, the measure family $C_{\alpha,1}$ internally possesses the same ordering, but for mixed states, the ordering relations can differ and may even reverse multiple times (Proposition 4.1). Similarly, ordering for mixed states can also be inconsistent among different measures (Proposition 4.2). Subsequently, we established a series of universal inequalities among different block coherence measures (Propositions 4.3-4.8).
	\par
	In Chapter 5, we embedded the S-T coherence measure $C_{\alpha,1}$ into the radical pair reaction master equation proposed by Kominis, where the recombination term is coupled to the strength of coherence. The results demonstrate the complex competition among the Hamiltonian generating coherence, decoherence destroying coherence, and the nonlinear feedback of the recombination term depending on coherence. This intuitively reflects the decisive influence of quantum coherence on the final chemical reaction product yield.
	
	\section*{Acknowledgement}
	The author Xiangyu Chen would like to thank Mengli Guo from the Department of Mathematics, East China University of Technology, for the helpful discussions on the Tsallis relative operator entropy.


\begin{thebibliography}{99}
		\bibitem{vedral1998entanglement}
		V. Vedral and M. B. Plenio.
		\newblock Entanglement measures and purification procedures
		\newblock {\em Physical Review A}, \textbf{57}, 1619(1998).
		
		\bibitem{Plenio2005AnIT}
		M. B. Plenio and S. Virmani.
		\newblock An introduction to entanglement measures.
		\newblock {\em Quantum Information and Computation}, \textbf{7}, 1-51(2005).
		
		\bibitem{Hordecki2009Quantum}
		R. Horodecki, P. Horodecki, M. Horodecki and K. Horodecki. 
		\newblock Quantum entanglement.
		\newblock {\em Review of Modern Physics}, \textbf{81}, 865(2009).
		
		\bibitem{baumgratz2014quantifying}
		T. Baumgratz, M. Cramer and M. B. Plenio.
		\newblock Quantifying coherence.
		\newblock {\em Physical Review Letters}, \textbf{113}, 140401(2014).
		
		\bibitem{hickey2018quantifying}
		A. Hickey and G. Gour.
		\newblock Quantifying the imaginarity of quantum mechanics.
		\newblock {\em Journal of Physics A: Mathematical and Theoretical}, \textbf{51}, 414009(2018).
		
		\bibitem{Howard2017application}
		V. Veitch, S. H. Mousavian, D. Gottesman, J. Emerson.
		\newblock The resource theory of stabilizer quantum computation.
		\newblock {\em New Journal of Physics}, \textbf{16}, 013009(2014).
		
		\bibitem{Gour2015resource}
		G. Gour, M. P. M\"uller, V. Narasimhachar, R. W. Spekkens, N. Y. Halpern.
		\newblock The resource theory of informational nonequilibrium in thermodynamics.
		\newblock {\em Physics Reports}, \textbf{583}, 1-58(2015).
		
		\bibitem{Xu2023quantifying}
		J. W. Xu.
		\newblock Quantifying the phase of quantum states.
		\newblock {\em Physics Letters A}, \textbf{482}, 129049(2023).
		
		\bibitem{Amaral2018Noncontextual}
		B. Amaral, A. Cabello, M. T. Cunha, L. Aolita.
		\newblock Noncontextual Wirings.
		\newblock {\em Physical Review Letters}, \textbf{120}, 130403(2018).
		
		\bibitem{xu2019coherence}
		J. W. Xu.
		\newblock Coherence of quantum channels.
		\newblock {\em Physical Review A}, \textbf{100}, 052311(2019).
		
		\bibitem{Luo2017partial}
		S. L. Luo and Y. Sun
		\newblock Partial coherence with application to the monotonicity problem of coherence involving skew information.
		\newblock {\em Physical Review A}, \textbf{96}, 022136(2017).
		
		\bibitem{aberg2006quantifying}
		J. Aberg.
		\newblock Quantifying superposition.
		\newblock {\em arXiv preprint quant-ph/0612146}, (2006).
		
		\bibitem{bischof2019resource}
		F. Bischof, H. Kampermann, D. Bruß.
		\newblock Resource theory of coherence based on positive-operator-valued measures.
		\newblock {\em Physical Review Letters}, \textbf{123}, 110402(2019).
		
		\bibitem{bischof2021quantifying}
		F. Bischof, H. Kampermann, D. Bruß.
		\newblock Quantifying coherence with respect to general quantum measurements.
		\newblock {\em Physical Review A}, \textbf{103}, 032429(2021).
		
		\bibitem{Yu2021quantum}
		C. S. Yu.
		\newblock Quantum coherence via skew information and its polygamy.
		\newblock {\em Physical Review A}, \textbf{95}, 042337(2017).
		
		\bibitem{kominis2025physiological}
		I. K. Kominis.
		\newblock Physiological search for quantum biological sensing effects based on the Wigner--Yanase connection between coherence and uncertainty.
		\newblock {\em Advanced Quantum Technologies}, \textbf{8}, 2300292(2025).
		
		\bibitem{kominis2015radical}
		I. K. Kominis.
		\newblock The radical-pair mechanism as a paradigm for the emerging science of quantum biology.
		\newblock {\em Modern Physics Letters B}, \textbf{29}, 1530013(2015).
		
		\bibitem{xu2020general}
		J. W. Xu, L. H. Shao and S. M. Fei.
		\newblock Coherence measures with respect to general quantum measurements.
		\newblock {\em Physical Review A}, \textbf{102}, 012411(2020).
		
		\bibitem{kritsotakis2014retrodictive}
		M. Kritsotakis and I. K. Kominis.
		\newblock Retrodictive derivation of the radical-ion-pair master equation and Monte Carlo simulation with single-molecule quantum trajectories.
		\newblock {\em Physical Review E}, \textbf{90}, 042719(2014).
		
		\bibitem{kominis2020quantum}
		I. K. Kominis.
		\newblock Quantum relative entropy shows singlet-triplet coherence is a resource in the radical-pair mechanism of biological magnetic sensing.
		\newblock {\em Physical Review Research}, \textbf{2}, 023206(2020).
		
		\bibitem{chen2023measures}
		Q. Chen, T. Gao and F. Yan.
		\newblock Measures of imaginarity and quantum state order.
		\newblock {\em Science China Physics, Mechanics \& Astronomy}, \textbf{66}, 280312(2023).
		
		\bibitem{liu2018superadditivity}
		C. L. Liu, Q. M. Ding and D. M. Tong.
		\newblock Superadditivity of convex roof coherence measures.
		\newblock {\em Journal of Physics A: Mathematical and Theoretical}, \textbf{51}, 414012(2018).
		
		\bibitem{yu2016alternative}
		X. D. Yu, D. J. Zhang, G. F. Xu and D. M. Tong.
		\newblock Alternative framework for quantifying coherence.
		\newblock {\em Physical Review A}, \textbf{94}, 060302(2016).
		
		\bibitem{xue2021quantification}
		S. N. Xue, J. Guo, P. Li, M. Ye and Y. Li.
		\newblock Quantification of resource theory of imaginarity.
		\newblock {\em Quantum Information Processing}, \textbf{20}, 383(2021).
		
		\bibitem{audenaert2015alpha}
		K. M. R. Audenaert, N. Datta.
		\newblock $\alpha$-$z$-Rényi relative entropies.
		\newblock {\em Journal of Mathematical Physics}, \textbf{56}, (2015).
		
		\bibitem{audenaert2007araki}
		K. M. R. Audenaert.
		\newblock On the araki-lieb-thirring inequality.
		\newblock {\em arXiv preprint math/0701129}, (2007).
		
		\bibitem{muller2013quantum}
		M. Müller-Lennert, F. Dupuis, O. Szehr, S. Fehr, M. Tomamichel.
		\newblock On quantum Rényi entropies: A new generalization and some properties.
		\newblock {\em Journal of Mathematical Physics}, \textbf{54}, (2013).
		
		\bibitem{bikchentaev2024trace}
		A. M. Bikchentaev, F. Kittaneh, M. S. Moslehian, Y. Seo.
		\newblock Trace inequalities.
		\newblock Springer, (2024).
		
		\bibitem{abe2003monotonic}
		S. Abe.
		\newblock Monotonic decrease of the quantum nonadditive divergence by projective measurements.
		\newblock {\em Physics Letters A}, \textbf{312}, 336--338(2003).
		
		\bibitem{lei2021povm}
		Q. Lei.
		\newblock POVM-induced coherence measure in terms of fidelity.
		\newblock {\em International Journal of Theoretical Physics}, \textbf{60}, 2423--2428(2021).
		
		\bibitem{fu2022block}
		L. Fu, F. Yan and T. Gao.
		\newblock Block-coherence measures and coherence measures based on positive-operator-valued measures.
		\newblock {\em Communications in Theoretical Physics}, \textbf{74}, 025104(2022).
		
		\bibitem{miranowicz2004ordering}
		A. Miranowicz and A. Grudka.
		\newblock Ordering two-qubit states with concurrence and negativity.
		\newblock {\em Physical Review A}, \textbf{70}, 032326(2004).
		
		\bibitem{liu2016ordering}
		C. L. Liu, X. D. Yu, G. F. Xu and D. M. Tong.
		\newblock Ordering states with coherence measures.
		\newblock {\em Quantum Information Processing}, \textbf{15}, 4189--4201(2016).
		
		\bibitem{carlen2018some}
		E. A. Carlen, E. H. Lieb.
		\newblock Some trace inequalities for exponential and logarithmic functions.
		\newblock {\em Bulletin of mathematical sciences}, 1--40(2018).
		
		\bibitem{zhang2016lower}
		L. Zhang, K. Bu, J. Wu.
		\newblock A lower bound on the fidelity between two states in terms of their trace-distance and max-relative entropy.
		\newblock {\em Linear and Multilinear Algebra}, \textbf{64}, 801--806(2016).
		
		\bibitem{datta2009min}
		N. Datta.
		\newblock Min-and max-relative entropies and a new entanglement monotone.
		\newblock {\em IEEE Transactions on Information Theory}, \textbf{55}, 2816--2826(2009).
		
		\bibitem{jones2010spin}
		J. A. Jones, P. J. Hore.
		\newblock Spin-selective reactions of radical pairs act as quantum measurements.
		\newblock {\em Chemical Physics Letters}, \textbf{488}, 90--93(2010).
		
	\end{thebibliography}
\end{document}